%% file: Real_World_Option_Pricing.tex
\documentclass[11pt,a4paper,notitlepage]{article}

\usepackage{authblk}
\usepackage{titlesec}

\input{Custom_Preamble.tex}

\usepackage{a4wide}	

\newcommand{\dGOP}{\bar{S}^*}
\newcommand{\GOP}{{S}^*}
\newcommand{\Ainfo}{\mathcal{A}}
\newcommand{\Afiltration}{\underline{\mathcal{A}}}

\newcommand{\vecW}{\mathbf{w}}

\theoremstyle{plain}

\newtheorem{defn}{Definition}[section]
\newtheorem{thm}{Theorem}[section]

\newtheorem{prop}[thm]{Proposition}

\title{Quantization Under the Real-world Measure: \\Fast and Accurate Valuation of Long-dated Contracts}
\author[1]{Ralph Rudd}
\author[1,2]{Thomas A. McWalter\thanks{Correspondence: tom@analytical.co.za}}
\author[1,3]{J\"{o}rg Kienitz}
\author[1,4]{Eckhard Platen}
\affil[1]{The African Institute for Financial Markets and Risk Management (AIFMRM), University of Cape Town}
\affil[2]{Department of Finance and Investment Management, University of Johannesburg}
\affil[3]{Fachbereich Mathematik und Naturwissenschaften,
Bergische Universit\"{a}t Wuppertal}
\affil[4]{Finance Discipline Group and School of Mathematical and Physical Sciences, University of Technology Sydney}

\begin{document}

\maketitle

\begin{abstract}
This paper provides a methodology for fast and accurate pricing of the long-dated contracts that arise as the building blocks of insurance and pension fund agreements.
It applies the recursive marginal quantization (RMQ) and joint recursive marginal quantization (JRMQ) algorithms outside the framework of traditional risk-neutral methods by pricing options under the real-world probability measure, using the benchmark approach. The benchmark approach is reviewed, and the real-world pricing theorem is presented and applied to various long-dated claims to obtain less expensive prices than suggested by traditional risk-neutral valuation. The growth-optimal portfolio (GOP), the central object of the benchmark approach, is modelled using the time-dependent constant elasticity of variance model (TCEV). Analytic European option prices are derived and the RMQ algorithm is used to efficiently and accurately price Bermudan options on the GOP. The TCEV model is then combined with a $3/2$ stochastic short-rate model and RMQ is used to price zero-coupon bonds and zero-coupon bond options, highlighting the departure from risk-neutral pricing.
\end{abstract}

\input{Sections/Introduction}

\input{Sections/RealWorldPricing}

\input{Sections/OptionPricing}

\input{Sections/Conclusion}

\clearpage
\bibliographystyle{abbrvnat}
\bibliography{/Projects/Bibliography/References}

\newpage
\appendix
\input{Sections/Appendix}

\end{document}

%% file: Custom_Preamble.tex
\usepackage[square]{natbib}
\usepackage{amsmath}
\usepackage{amsthm}
\usepackage{amsfonts}                                 	
\usepackage{pgf}                                     	
\usepackage{mathrsfs}									
\usepackage{mathtools}									
\usepackage{hyperref}
\hypersetup{
    colorlinks=true,  
    linkcolor=blue,   
    citecolor=blue
}

\makeatletter
\let\save@mathaccent\mathaccent
\newcommand*\if@single[3]{%
  \setbox0\hbox{${\mathaccent"0362{#1}}^H$}%
  \setbox2\hbox{${\mathaccent"0362{\kern0pt#1}}^H$}%
  \ifdim\ht0=\ht2 #3\else #2\fi
  }
\newcommand*\rel@kern[1]{\kern#1\dimexpr\macc@kerna}
\newcommand*\widebar[1]{\@ifnextchar^{{\wide@bar{#1}{0}}}{\wide@bar{#1}{1}}}
\newcommand*\wide@bar[2]{\if@single{#1}{\wide@bar@{#1}{#2}{1}}{\wide@bar@{#1}{#2}{2}}}
\newcommand*\wide@bar@[3]{%
  \begingroup
  \def\mathaccent##1##2{%
    \let\mathaccent\save@mathaccent
    \if#32 \let\macc@nucleus\first@char \fi
    \setbox\z@\hbox{$\macc@style{\macc@nucleus}_{}$}%
    \setbox\tw@\hbox{$\macc@style{\macc@nucleus}{}_{}$}%
    \dimen@\wd\tw@
    \advance\dimen@-\wd\z@
    \divide\dimen@ 3
    \@tempdima\wd\tw@
    \advance\@tempdima-\scriptspace
    \divide\@tempdima 10
    \advance\dimen@-\@tempdima
    \ifdim\dimen@>\z@ \dimen@0pt\fi
    \rel@kern{0.6}\kern-\dimen@
    \if#31
      \overline{\rel@kern{-0.6}\kern\dimen@\macc@nucleus\rel@kern{0.4}\kern\dimen@}%
      \advance\dimen@0.4\dimexpr\macc@kerna
      \let\final@kern#2%
      \ifdim\dimen@<\z@ \let\final@kern1\fi
      \if\final@kern1 \kern-\dimen@\fi
    \else
      \overline{\rel@kern{-0.6}\kern\dimen@#1}%
    \fi
  }%
  \macc@depth\@ne
  \let\math@bgroup\@empty \let\math@egroup\macc@set@skewchar
  \mathsurround\z@ \frozen@everymath{\mathgroup\macc@group\relax}%
  \macc@set@skewchar\relax
  \let\mathaccentV\macc@nested@a
  \if#31
    \macc@nested@a\relax111{#1}%
  \else
    \def\gobble@till@marker##1\endmarker{}%
    \futurelet\first@char\gobble@till@marker#1\endmarker
    \ifcat\noexpand\first@char A\else
      \def\first@char{}%
    \fi
    \macc@nested@a\relax111{\first@char}%
  \fi
  \endgroup
}
\makeatother

\newcommand{\R}{\mathbb{R}}                         	
\renewcommand{\P}{\mathbb{P}}                       	

\newcommand{\ind}[1]{\mathbb{I}_{\left\{#1\right\}}}	

\DeclarePairedDelimiter\norm{\lvert}{\rvert}

\DeclarePairedDelimiter\Ebraces{[}{]}
\newcommand{\E}{\mathbb{E}\Ebraces}
\newcommand{\abs}{\norm}



%% file: Sections/Introduction.tex
\section{Introduction}
    
It has been argued that requiring the existence of a risk-neutral measure is an unrealistic modelling constraint, severe enough to limit the efficacy of long-term pricing and hedging strategies \citep{hulley2012hedging}.

It is, however, possible to derive a unified framework for portfolio optimization, derivative pricing and risk management without the need for an equivalent martingale measure, as described by \citet{platen2004benchmark}, the seminal work on the \emph{benchmark approach}. Under this framework, the Law of One Price no longer holds \citep{platen2008law}, and the Fundamental Theorems of Asset Pricing \citep{delbaen1994general,delbaen1998fundamental} do not apply,  as the traditional no-arbitrage concept is relaxed. As a consequence, certain long-dated contracts can be less expensively replicated than suggested by the classical theory.

Central to the benchmark approach is the concept of the growth-optimal portfolio (GOP), first explored by \citet{kelly1956new}. This portfolio strategy is such that, when denoted in units of the GOP, every asset becomes a supermartingale. This allows the derivation of the real-world pricing theorem; yielding the least expensive prices under the real-world probability measure.

The goal of this work is to provide a simple numerical toolbox for the pricing of long-dated contracts under the benchmark approach by using recursive marginal quantization (RMQ). RMQ was introduced by \cite{pages2015recursive} and is a numerical technique for the optimal approximation of functionals of solutions to stochastic differential equations (SDEs).

The main contribution of this work is two-fold. First, analytic European option prices are derived for the time-dependent constant elasticity of variance (TCEV) model for the GOP. This is an extension of formulae presented by \cite{baldeaux2014tractable}, which generalize the previous option pricing formulae from \cite{miller2008analytic, miller2010real}. Secondly, RMQ is shown to be a highly effective tool for pricing both European options on the GOP, under the assumption of constant interest rates, and zero-coupon bond options, under a stochastic short-rate model, namely the $3/2$-model from \cite{ahn1999parametric}. This efficient pricing mechanism allows for the wider application of the benchmark approach.

The paper proceeds as follows.
In Section \ref{Sec: Benchmark Approach} the benchmark approach is briefly reviewed in the context of a two-asset scalar diffusion market. The form of the growth-optimal portfolio strategy is specified, and the real-world pricing theorem is derived. The TCEV model is introduced and its probabilistic features are detailed.

In the first part of Section \ref{Sec: Option Pricing}, analytic European option pricing formulae are derived for the TCEV model under the assumption of a constant interest rate. The pricing efficiency of these formulae is compared with the approximate prices obtained using RMQ. The latter part of Section \ref{Sec: Option Pricing} deals with stochastic interest rates, specifically the hybrid model introduced by \cite{baldeaux2015hybrid}. Analytic zero-coupon bond prices are presented under the assumption of independence between the stochastic short rate and the GOP. The influence of correlation is briefly explored. Finally, RMQ is used to efficiently price European options on zero-coupon bonds in this framework. Section \ref{Sec: Conclusion} concludes.

%% file: Sections/RealWorldPricing.tex
\section{The Benchmark Approach}
\label{Sec: Benchmark Approach}
This section presents a one-dimensional version of the continuous financial market presented by \citet{platen2004benchmark}, similar to the introduction given by \citet[Chap.~9]{platen2006benchmark}. For full mathematical rigor see the original derivation of real-world pricing by \cite{platen2002arbitrage}, and the extension to the jump-diffusion framework \citep[Chap.~14]{platen2006benchmark}. The most general derivation of the ``Law of Minimal Price'', of which real-world pricing is a result, is presented by \citet{platen2008law}.

Assume the existence of a filtered probability space, $(\Omega, \Ainfo, \Afiltration, \P )$. The filtration $\Afiltration = (\Ainfo_t)_{t\in[0,\infty)}$ is assumed to satisfy the usual conditions.

Consider a continuous financial market consisting of two assets: a risk-free savings account, $S^0 = \{S^0_t,\ t\in[0,\infty)\}$, and a risky primary security, $S^1 = \{S^1_t,\ t\in[0,\infty)\}$. They respectively obey the SDEs
\begin{align}
dS^0_t &= S^0_tr_t\,dt \label{Eq: Savings Account}\\
\intertext{and}
dS^1_t &= S^1_t(a_t\,dt+b_t\,dW_t), \label{Eq: Diffusion}
\end{align}
with $r = \{r_t,\ t\in[0,\infty)\}$, the adapted short rate process. 
Here $b = \{b_t,\ t\in[0,\infty)\}$ is a predictable and strictly positive process known as the diffusion coefficient of $S^1$ with respect to the standard Brownian motion, $W$, and is assumed to satisfy
\begin{equation*}
    \int_0^T b_s^2\,ds < \infty
\end{equation*}
almost surely for $T\in[0,\infty)$, a finite time-horizon. It is also assumed that $a = \{a_t,\ t\in[0,\infty)\}$, known as the drift, is a predictable process satisfying
\begin{equation*}
    \int_0^T |a_s|\,ds < \infty,
\end{equation*}
almost surely. It is assumed that the SDE \eqref{Eq: Diffusion} has a unique strong solution, see, for example, \citet{platen2006benchmark}. 

In the market considered, the number of risky securities is the same as the number of Wiener processes. Refer to \cite{platen2004pricing} for the case where there are more sources of uncertainty than traded securities.

Defining the market price of risk (MPOR) process as
\begin{equation*}
    \theta_t \vcentcolon= b_t^{-1}(a_t-r_t),
\end{equation*}
allows the SDE \eqref{Eq: Diffusion} to be re-written as
\begin{equation}
    dS^1_t = S^1_t\left( (r_t + b_t\theta_t)\,dt + b_t\,dW_t) \right). \label{Eq: Diffusion 2}
\end{equation}
It is assumed that the absolute value of the MPOR process is always finite. 

\subsection{The Growth Optimal Portfolio}
The GOP has often been attributed to \citet{kelly1956new}. \citet{hakansson1995capital} state in their review of the growth optimal investment strategy that the GOP was already implied by Bernoulli's solution to the St.~Petersburg Paradox as early as 1738 (see \cite{samuelson1977st} for this interesting digression). 
An important application of the GOP to claim valuation and long-run portfolio growth is \cite{long1990numeraire}, where it is shown that, under certain constraints, risk-neutral pricing is equivalent to pricing under the real-world probability measure using the GOP as the numeraire.

The GOP is the central object of the benchmark approach and hence real-world pricing. 
In a general semimartingale framework, \cite{karatzas2007numeraire} show that the ``no unbounded profit with bounded risk'' no-arbitrage condition is necessary and sufficient for the existence of the GOP. This condition is placed into its proper context alongside the range of applicable arbitrage statements by \cite{fontana2015weak}.

In the market defined above, a predictable stochastic vector process $\delta = \{\delta_t = (\delta_t^0, \delta_t^1) ,\ t\in[0,T]\}$ is called a \emph{strategy} if, for all $t\in[0,T]$, the stochastic It\^{o} integrals
\begin{equation*}
\int_0^t\delta^0_s\, dS^0_s \qquad\mathrm{and}\qquad \int_0^t\delta^1_s\,dS^1_s
\end{equation*}
exist.
Denote by $S^\delta = \{S^\delta_t, t\in[0,T]\}$ the time-$t$ value of the portfolio process associated with the strategy $\delta$, defined as 
\begin{equation*}
    S^\delta_t = \delta^0_t S^0_t + \delta^1_t S^1_t.
\end{equation*}
A strategy and its corresponding portfolio process are said to be self-financing if
\begin{equation*}
    dS^\delta_t = \delta^0_t\,dS^0_t + \delta^1_t\,dS^1_t,
\end{equation*}
for all $t\in[0,T]$.
The self-financing condition ensures that instantaneous changes in the value of the portfolio are due to changes in the prices of the constituent securities and not to external deposits or withdrawals. Only self-financing portfolios are considered in the following.

At this point it is necessary to introduce the concept of \emph{admissable} portfolios, to avoid the arbitrage opportunities generated by traditional doubling strategies. Admissable strategies are usually either constrained via an absolute lower bound or an integrability condition (see \citet[Chap.~7]{hunt2004financial} for a discussion).
Only the set, $\mathcal{V}^+$, of strictly positive portfolios will be considered, thus providing an absolute lower bound at zero.
For $S^\delta\in\mathcal{V}^+$, define the portfolio fraction
\begin{equation*}
    \pi^j_{\delta, t} = \frac{\delta^j_tS^j_t}{S^\delta_t},
\end{equation*}
as the fraction of the total portfolio value invested in each asset, $S^j_t$, for $j\in\{0,1\}$. Portfolio fractions can be negative but must always sum to one. Using \eqref{Eq: Diffusion 2}, the SDE for a self-financing portfolio in $\mathcal{V}^+$ can now be written as
\begin{equation*}
    dS^\delta_t = S^\delta_t\left((r_t + \pi_{\delta,t}^1b_t\theta_t)\,dt + \pi_{\delta,t}^1b_t\,dW_t\right).
\end{equation*}
A simple application of It\^{o}'s formula provides the SDE for the logarithm of the portfolio 
\begin{equation*}
    d\log S^\delta_t = g_{\delta,t}\,dt + \pi_{\delta,t}^1b_t\,dW_t,
\end{equation*}
with portfolio growth rate
\begin{equation}
    g_{\delta,t} = r_t + \pi_{\delta,t}^1b_t\theta_t - \frac{1}{2}(\pi^1_{\delta, t}b_t)^2. \label{Eq: Growth Rate}
\end{equation}

The GOP is the portfolio that maximises this growth rate, that is, the drift of the log-portfolio.
Mathematically, a strictly positive portfolio value process, $S^{\delta_*} = \{S^{\delta_*}_t, t\in[0,T]\}$, is a growth optimal portfolio if, for all $t\in[0,T]$ and all $S^\delta\in\mathcal{V}^+$, the inequality
\begin{equation*}
    g_{\delta_*,t} \geq g_{\delta,t}
\end{equation*}
holds almost surely.
From the first-order condition for the maximum of the growth rate \eqref{Eq: Growth Rate}, the optimal fraction to be invested in $S^1$ can be found as
\begin{equation*}
    \pi_{\delta_*,t}^1 = b_t^{-1}\theta_t.
\end{equation*}
Consequently the SDE for the GOP is given by
\begin{equation}
    \label{Eq: The GOP SDE}
    dS^{\delta_*}_t = S_t^{\delta_*}((r_t+\theta^2_t)\,dt + \theta_t\,dW_t),
\end{equation}
with $t\in[0,T]$ and $S_0^{\delta_*} > 0.$ Contingent claims can now be priced under the real-world probability measure using the GOP as the numeraire or \emph{benchmark}, as will be detailed below.

\subsection{Real-world Pricing}
\label{Sec: Real-world Pricing}
Using the GOP as benchmark and numeraire, consider the evolution of a benchmarked portfolio given by the ratio
\begin{equation*}
    \hat{S}^\delta_t = \frac{S^\delta_t}{S^{\delta_*}_t}.
\end{equation*}
It\^{o}'s formula provides the SDE
\begin{equation}
    d\hat{S}^\delta_t = (\delta^1_t\hat{S}^1_tb_t - \hat{S}^\delta_t\theta_t)\,dW_t, \label{Eq: Benchmarked Portfolio}
\end{equation}
in terms of $\hat{S}_t^1 = \frac{S^1_t}{S^{\delta_*}_t}$, the benchmarked security price process.

Because no drift is present in \eqref{Eq: Benchmarked Portfolio}, it is clear that benchmarked portfolios form $(\Afiltration,\P)$-local martingales. Thus, by Fatou's lemma, all non-negative portfolios, when benchmarked, are $(\Afiltration,\P)$-supermartingales\footnote{For a simple proof of this classic result, see \citet[Lemma 5.2.3]{platen2006benchmark}. }.

Define a non-negative contingent claim, $V_\tau$, that matures at a stopping time $\tau\in[0,T]$, as an $\Ainfo_\tau$-measurable payoff that possesses a finite expectation when benchmarked. Note that $V_\tau$ need not be square-integrable. Let $S^V_t$ denote a non-negative self-financing portfolio that replicates the claim, i.e., $S^V_\tau = V_\tau$. Then
\begin{equation*}
    \frac{S^V_t}{S^{\delta_*}_t} \geq \E*{\frac{V_T}{S^{\delta_*}_T}\Bigl|\Ainfo_t}
\end{equation*}
holds by the supermartingale property of benchmarked non-negative self-financing  portfolios.

A security price process, equivalent to a self-financing, replicating portfolio, is called \textit{fair} if its benchmarked value forms an $(\Afiltration,\P)$-martingale (in the classical risk-neutral setting, this notion of a fair process is equivalent to that proposed by \cite{geman1995changes}). Under the benchmark approach this leads to the minimal possible price, the desired result for this section.
\begin{thm}[Real-world Pricing]
    \label{Thm: Real-world pricing}
    For any fair security price process, $V = \{V_t,\ t\in[0,T]\}$, $T\in(t,\infty)$, one has the real-world pricing formula
    \begin{equation}
        V_t = S_t^{\delta_*}\E*{\frac{V_T}{S^{\delta_*}_T}\Bigl|\Ainfo_t}.\label{Eq: Real-world pricing}
    \end{equation}
\end{thm}
The expectation in Theorem \ref{Thm: Real-world pricing} is taken under the real-world probability measure, $\P$. Under the assumption that one can perform an equivalent probability measure change, following \cite{geman1995changes}, the candidate Radon-Nikodym derivative process to move from the current real-world numeraire-measure pair, $(S^{\delta_*},\P)$, to the risk-neutral numeraire-measure pair, $(S^0, \P_\theta)$, is
\begin{equation}
    \lambda_\theta(t) = \frac{d\P_\theta}{d\P} \Bigl|_{\Ainfo_t} = \frac{S^0_t}{S^0_0}\frac{S^{\delta_*}_0}{S^{\delta_*}_t} = \frac{\hat{S}^0_t}{\hat{S}^0_0}. \label{Eq: RWP Candidate Measure Change}
\end{equation}
If $\lambda_\theta(t)$ is a strictly positive $(\Afiltration,\P)$-martingale (and not a strict local-martingale), then the probability measure change can indeed be performed, yielding
\begin{equation*}
    V_t = S_t^{\delta_*}\E*{\frac{V_T}{S^{\delta_*}_T}\Bigl|\Ainfo_t} = \E*{\frac{\lambda_\theta(T)}{\lambda_\theta(t)} \frac{S^0_t}{S^0_T}V_T\Bigl|\Ainfo_t} = \mathbb{E}_\theta{\left[\frac{S^0_t}{S^0_T}V_T\Bigl|\Ainfo_t\right]}, 
\end{equation*}
by Bayes' theorem, where the last expression is the classical risk-neutral pricing formula.

In this way, risk-neutral pricing is a special case of real-world pricing, applicable only when $\lambda_\theta(t)$ describes a strictly positive $(\Afiltration,\P)$-martingale. This translates into a constraint on the market price of risk $\theta_t$, the volatility of the GOP. This volatility must be specified so that $\lambda_\theta(t)$ is a martingale, which is, for instance, the case if $\theta_t$ satisfies Novikov's condition or, more generally, Kazamaki's condition \citep{revuz2013continuous}. To compute the expectation in Theorem \ref{Thm: Real-world pricing}, the GOP must be modelled explicitly. 
A realistic model for the GOP, which excludes risk-neutral pricing, is presented in the next section.

\subsection{Modelling the GOP}
Consider a simple two-asset market consisting only of the risk-free bank account and the growth-optimal portfolio, obeying \eqref{Eq: Savings Account} and \eqref{Eq: The GOP SDE}, respectively. The \emph{discounted} GOP is then described by
\begin{align*}
 \dGOP_t &= \frac{\GOP_t}{S^0_t}, \\
\intertext{which means that}
d\dGOP_t &= \dGOP_t(\theta_t^2\, dt + \theta_t\, dW_t).
\end{align*}
The $\delta_*$-superscript has been dropped from the GOP notation for simplicity. To compute the expectation in \eqref{Eq: Real-world pricing}, the MPOR process, $\theta_t$, must be modelled explicitly.

\cite{baldeaux2014tractable} propose the TCEV model for the GOP. It is parsimonious, tractable, reliably estimated and can provide explicit formulae for various derivatives and their hedge ratios. The TCEV model is specified by
\[ \theta_t \vcentcolon= c\left(\frac{\dGOP_t}{\alpha_t}\right)^{a - 1} \qquad \mathrm{and} \qquad \alpha_t \vcentcolon= \alpha_0e^{\eta t}. \]
Here the parameter restrictions are $c>0,\ a\in(-\infty,1),\ \alpha_0>0$ and $\eta>0$. Note that the TCEV model generalizes both the minimal market model (MMM), first introduced by \citet{Platen2001minimal} and analyzed in detail by \citet{miller2008analytic}, and the modified constant elasticity of variance model (MCEV) from \citet{heath2002consistent}.

By direct substitution, the SDE for the discounted GOP under the TCEV model is
\begin{equation}
d\dGOP_t = c^2\alpha_t^{2(1 - a)}\left(\dGOP_t\right)^{2a - 1}\, dt + c\alpha_t^{1-a}\left(\dGOP_t\right)^a\, dW_t. \label{Eq: DGOP TCEV}
\end{equation}
The behaviour of this process is described by Proposition \ref{Prop: DGOP to BESQ} below, which appears in a slightly modified form in \cite{baldeaux2014tractable}.

\begin{prop}
    \label{Prop: DGOP to BESQ}
    The process $\dGOP = \{\dGOP_t,\ t\geq0\}$ satisfies the following equality in distribution
    \begin{equation*}
    \dGOP_t \overset{(d)}{=}  X_{\varphi(t)}^{\frac{1}{2(1 - a)}} = X_{\varphi(t)}^{\left(\frac{\delta}{2} - 1\right)},
    \end{equation*}
    where $X = \{X_\varphi,\ \varphi\geq0\}$ is a squared Bessel process of dimension $\delta = \frac{3 - 2a}{1 - a}$ in $\varphi$-time and the time-transformation is given by
    \begin{equation*}
    \varphi(t) = \frac{(1 - a)\alpha_0^{2(1-a)}c^2}{2\eta}\left(e^{2(1 - a)\eta t} - 1\right).
    \end{equation*}
\end{prop}

\begin{proof}
Define $Y_t= \left(\tfrac{\alpha_0}{\alpha_t}\dGOP_t\right)^{2(1 - a)}$ such that
\begin{align*}
dY_t &= \left(\alpha_0^{2(1 - a)}c^2(1 - a)(3 - 2a) - 2(1 - a)\eta Y_t\right)\, dt + 2(1 - a)\alpha_0^{1 - a}c\sqrt{Y_t}\, dW_t.
\end{align*}
This is a CIR-process and thus by Proposition \ref{Prop: CIR to BESQ}, in Appendix \ref{App: Squared Bessel Processes},
\[ Y_t = e^{-2(1 - a)\eta t}X_{\varphi({t})}.\]
Since $\dGOP_t = e^{\eta t}Y_t^{\frac{1}{2(1-a)}},$ this completes the proof.
\end{proof}

An immediate consequence of the relationship between the discounted GOP and a squared Bessel process (BESQ) of dimension $\delta = \frac{3 - 2a}{1 - a} > 2$, is that the discounted GOP never attains zero (see Appendix \ref{Sec: BEQ Transition Density}). A more subtle consequence is that modelling the GOP portfolio in this way precludes the existence of an equivalent risk-neutral probability measure.

As in Section \ref{Sec: Real-world Pricing}, the candidate Radon-Nikodym derivative process to move to the risk-neutral measure is given by \eqref{Eq: RWP Candidate Measure Change}.
However, now
\begin{equation*}
\hat{S}^0_t = \frac{1}{\dGOP_t} = X_{\varphi(t)}^{1 - \frac{\delta}{2}},
\end{equation*}
and from the symmetry relationship derived for BESQ processes in Appendix \ref{Sec: BEQ Transition Density}, the process on the right-hand side of the above expression is a \emph{strict} local martingale. Thus, the `risk-neutral measure' induced by this Radon-Nikodym derivative process will not be a probability measure.

Alternatively, consider the existence of a measure, $\P_\theta$, such that the discounted GOP is potentially a martingale under this measure,
\begin{equation}
d\dGOP_t =  c\alpha_t^{1-a}\left(\dGOP_t\right)^a\, dW^\theta_t. \label{Eq: RN dGOP}
\end{equation}
Using the same transformation as in the proof of Proposition \ref{Prop: DGOP to BESQ}, $Y_t= \left(\tfrac{\alpha_0}{\alpha_t}\dGOP_t\right)^{2(1 - a)}$, it is straightforward to show that under this hypothetical measure the discounted GOP is the power of a BESQ process of dimension $\delta_\theta = \tfrac{1 - 2a}{1 - a} < 2$. As this process has a non-zero probability of attaining zero in finite time, the measure $\P_\theta$ cannot be equivalent to $\P$, the original real-world probability measure under which $X$ has dimension $\delta = \frac{3 - 2a}{1 - a} > 2$ and never hits zero.

%% file: Sections/OptionPricing.tex
\section{Option Pricing}
\label{Sec: Option Pricing}
In this section, recursive marginal quantization is used to provide fast and accurate pricing for the benchmark approach.
RMQ was introduced by \cite{pages2015recursive}, and extended to higher-order schemes by \cite{mcwalter2018}. 

Initially, analytic European option prices are derived under the assumption of constant interest rates. These formulae generalize those found in \cite{miller2008analytic,miller2010real} for the minimal market and modified constant elasticity of variance models, respectively. Although these formulae are analytic, they can be numerically expensive to compute and are contrasted to the fast, but approximate, prices obtained via RMQ. Furthermore, Bermudan options on the GOP are priced using traditional Monte Carlo methods, with their accuracy, speed and efficiency compared with that of recursive marginal quantization.

The second subsection deals with the hybrid model introduced by \cite{baldeaux2015hybrid}. The hybrid model combines the TCEV model for the GOP with a 3/2 stochastic short-rate model. \cite{baldeaux2015hybrid} derive analytic zero-coupon bond prices under the assumption that the GOP is independent of the short rate. In this section, numerical experiments show that these prices are well approximated with RMQ. The effect of the independence assumption on the zero-coupon bond prices is also investigated. Lastly, European options on zero-coupon bonds are priced using both traditional Monte Carlo methods as well as RMQ.


All simulations were performed using MATLAB 2016b on a computer with a $2.00$ GHz Intel i-$3$ processor and $4$ GB of RAM. 

\subsection{Constant Short Rates}

Expressions similar to those derived in Propositions \ref{Prop: Put Option} and \ref{Prop: Call Option} below appear in \cite{baldeaux2014tractable}, where the strike was selected to be a constant multiple of the savings account. In this way, \cite{baldeaux2014tractable} were able to avoid specifying a model for the short rate by restricting the class of strikes they considered.

\begin{prop}
    \label{Prop: Put Option}
    Assuming a constant interest rate $r$, the real-world price, $p_{T,K}(t, \GOP_t)$, of a European put option on the GOP at time $t$ with expiry $T$ and strike $K$ is given by
    \begin{align*}
    p_{T,K}(t, \GOP_t) &= -\dGOP_t \beta(t) \chi'^{\,2}\left(\frac{\widetilde{K}}{\Delta \varphi}; \delta,\frac{(\dGOP_t)^{2(1 - a)}}{\Delta \varphi}  \right) \\
      &\qquad + K \frac{\beta(t)}{\beta(T)}\left[\chi'^{\,2}\left(\frac{(\dGOP_t)^{2(1 - a)}}{\Delta \varphi}; \delta - 2, 0\right) - \chi'^{\,2}\left(\frac{(\dGOP_t)^{2(1 - a)}}{\Delta \varphi}; \delta - 2, \frac{\widetilde{K}}{\Delta \varphi}\right) \right],
    \end{align*}
    where
    \begin{align*}
    \widetilde{K} &= \left(\frac{K}{\beta(T)}\right)^{2(1 - a)},& \beta(t) &= \exp(rt),\\
    \delta &= \frac{3-2a}{1 - a},& \Delta\varphi&= \varphi(T) - \varphi(t),
    \end{align*}
    with $\chi'^{\,2}(x; \delta, \lambda)$ representing the noncentral chi-squared distribution, evaluated at $x$ with degrees of freedom $\delta$ and non-centrality parameter $\lambda$, and $\varphi(t)$ is defined in Proposition \ref{Prop: DGOP to BESQ}.
\end{prop}
\begin{proof}
    When the short rate is constant, the savings account is deterministic, with
    \[S^0_t = \exp(rt) =\vcentcolon \beta(t) . \]
    As is standard, the expectation of the numeraire denominated payoff can be expressed as the difference of two expectations,
    \begin{align*}
    p_{T,K}(t, \GOP_t) &= \E*{\frac{\GOP_t}{\GOP_T}(K - \GOP_T)^+\Bigl|{\Ainfo_t}} \\
    &= \E*{\frac{\dGOP_t}{\dGOP_T} \beta(t) \left(\frac{K}{\beta(T)} - \dGOP_T\right)^+\Bigl|{\Ainfo_t}} \\
    &= \dGOP_t \frac{\beta(t)}{\beta(T)} K \E*{\frac{1}{\dGOP_T}\ind{\frac{K}{\beta(T)}> \dGOP_T} \Bigl|{\Ainfo_t}}- \dGOP_t\beta(t)\E*{\ind{\frac{K}{\beta(T)}> \dGOP_T} \Bigl|{\Ainfo_t}}.
    \end{align*}

    Using Proposition \ref{Prop: DGOP to BESQ}, the first expectation can be rewritten in terms of the power of a squared Bessel process of dimension $\delta = \frac{3 - 2a}{1 - a} > 2$,
    \begin{align*}
    \E*{\frac{1}{\dGOP_t}\ind{\frac{K}{\beta(T)}> \dGOP_T} \Bigl|{\Ainfo_t}  }&= \E*{X^{-\left({\frac{\delta}{2}-1}\right)}_{\varphi(T)} \ind{\widetilde{K} > X_{\varphi(T)}} \Bigl|{\Ainfo_t}} \\
    &= \int_0^{\widetilde{K}} X^{-\left({\frac{\delta}{2}-1}\right)} p_{\delta>2}(X, \varphi(T); X_{\varphi(t)})\, dX,
    \end{align*}
    with the transition density, $p_{\delta>2}$, given by \eqref{Eq: Density dg2} in Appendix \ref{App: Squared Bessel Processes}. Now, the symmetry relationship, \eqref{Eq: Symmetry} in Appendix \ref{App: Squared Bessel Processes}, can be applied to yield
    \begin{align*}
    \int_0^{\widetilde{K}} X^{-\left({\frac{\delta}{2}-1}\right)} p_{\delta>2}(X, \varphi(T); X_{\varphi(t)})\, dX &= X^{-\left({\frac{\delta}{2}-1}\right)}_{\varphi(t)} \int_0^{\widetilde{K}}  p_{4 - \delta}(X, \varphi(T); X_{\varphi(t)})\, dX,
    \end{align*}
    where the final density to be integrated is the norm-decreasing density given by \eqref{Eq: Density dl2}. 
    The bounds can be rewritten as
    \begin{align*}
    &X^{-\left({\frac{\delta}{2}-1}\right)}_{\varphi(t)} \int_0^{\widetilde{K}}  p_{4 - \delta}(X, \varphi(T); X_{\varphi(t)})\, dX = \\
    &\qquad\qquad\qquad \frac{1}{\dGOP_t}\left[ \int_0^{\infty}  p_{4 - \delta}(X, \varphi(T); X_{\varphi(t)})\, dX  - \int_{\widetilde{K}}^\infty  p_{4 - \delta}(X, \varphi(T); X_{\varphi(t)})\, dX\right],
    \end{align*}
    and using \eqref{Eq: Schroder} and \eqref{Eq: Staying Pos Prob} from the Appendix gives
    \begin{equation*}
    \E*{\frac{1}{\dGOP_t}\ind{\frac{K}{\beta(T)}> \dGOP_T} \Bigl|{\Ainfo_t}} = \frac{1}{\dGOP_t}\left[\chi'^{\,2}\left(\frac{(\dGOP_t)^{2(1 - a)}}{\Delta \varphi}; \delta - 2, 0\right) - \chi'^{\,2}\left(\frac{(\dGOP_t)^{2(1 - a)}}{\Delta \varphi}; \delta - 2, \frac{\widetilde{K}}{\Delta \varphi}\right) \right].
    \end{equation*}
    Similarly, the second expectation can be computed directly from the transition density \eqref{Eq: Density dg2},
    \begin{align*}
    \E*{\ind{\frac{K}{\beta(T)}> \dGOP_T} \Bigl|{\Ainfo_t} } &= \int_0^{\widetilde{K}} p_{\delta>2}(X, \varphi(T); X_{\varphi(t)})\, dX \\
    &= \chi'^{\,2}\left(\frac{\widetilde{K}}{\Delta \varphi}; \delta,\frac{(\dGOP_t)^{2(1 - a)}}{\Delta \varphi}  \right).
    \end{align*}


\end{proof}

The analytic expression for the European call option can be derived in the same way, but the derivation avoids the norm-decreasing density required above.
\begin{prop}
    \label{Prop: Call Option}
    Assuming a constant interest rate $r$, the real-world price, $c_{T,K}(t, \GOP_t)$, of a European call option on the GOP at time $t$ with expiry $T$ and strike $K$ is given by
    \begin{align*}
    c_{T,K}(t, \GOP_t) = \dGOP_t \beta(t) \left[1 - \chi'^{\,2}\left(\frac{\widetilde{K}}{\Delta \varphi}; \delta, \frac{(\dGOP_t)^{2(1 - a)}}{\Delta \varphi}\right) \right]
      - K \frac{\beta(t)}{\beta(T)}\chi'^{\,2}\left(\frac{(\dGOP_t)^{2(1 - a)}}{\Delta \varphi}; \delta - 2,  \frac{\widetilde{K}}{\Delta \varphi}\right),
    \end{align*}
    with all definitions as in Proposition \ref{Prop: Put Option}.    
\end{prop}

\begin{figure}[t]
    \begin{center}
        \includegraphics[width=\columnwidth]{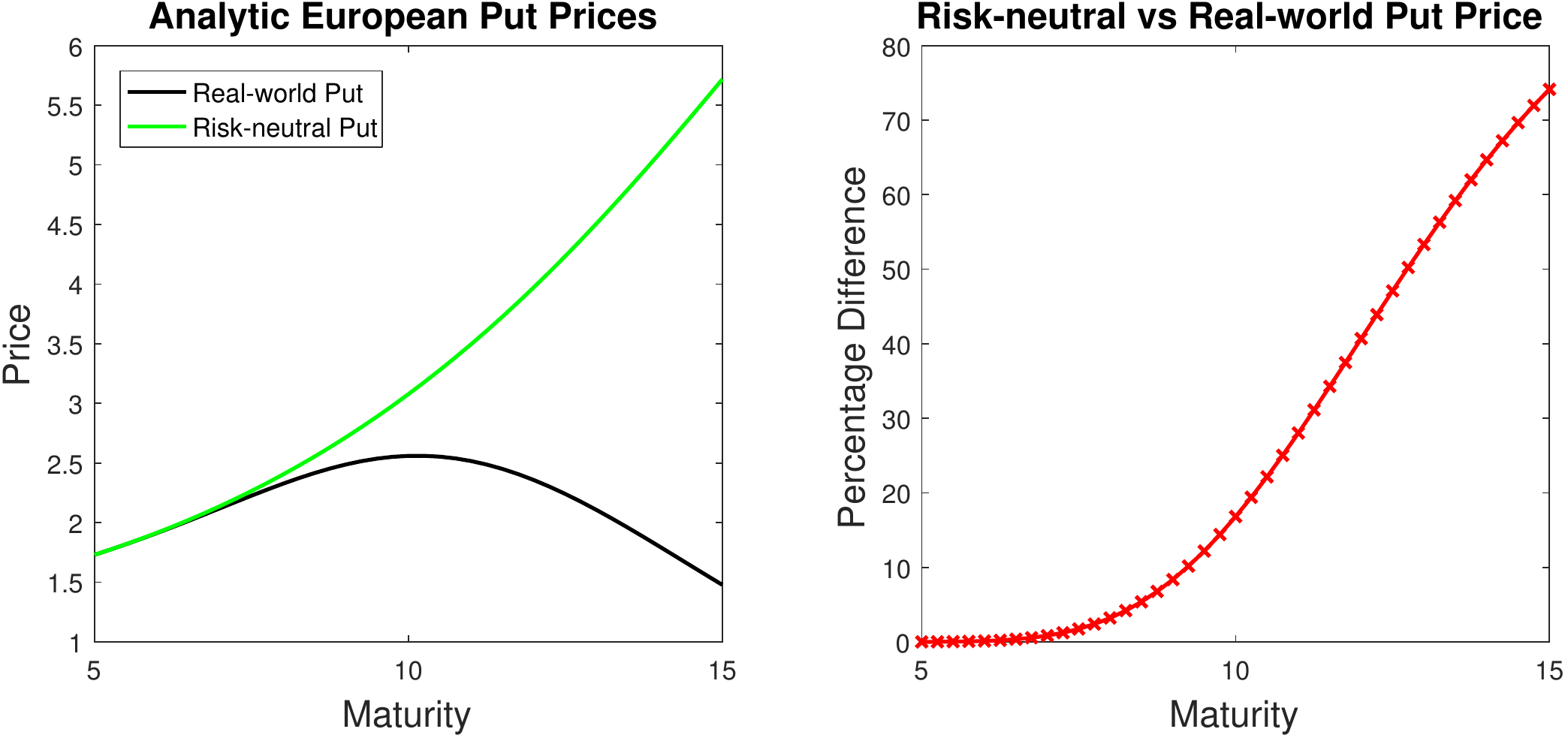}
    \end{center}
   \caption{Comparison of risk-neutral and real-world prices obtained for at-the-money European put options with maturities out to $15$ years.}
    \label{Fig: European put options}
\end{figure}

\begin{figure}[t]
    \begin{center}
        \includegraphics[width=\columnwidth]{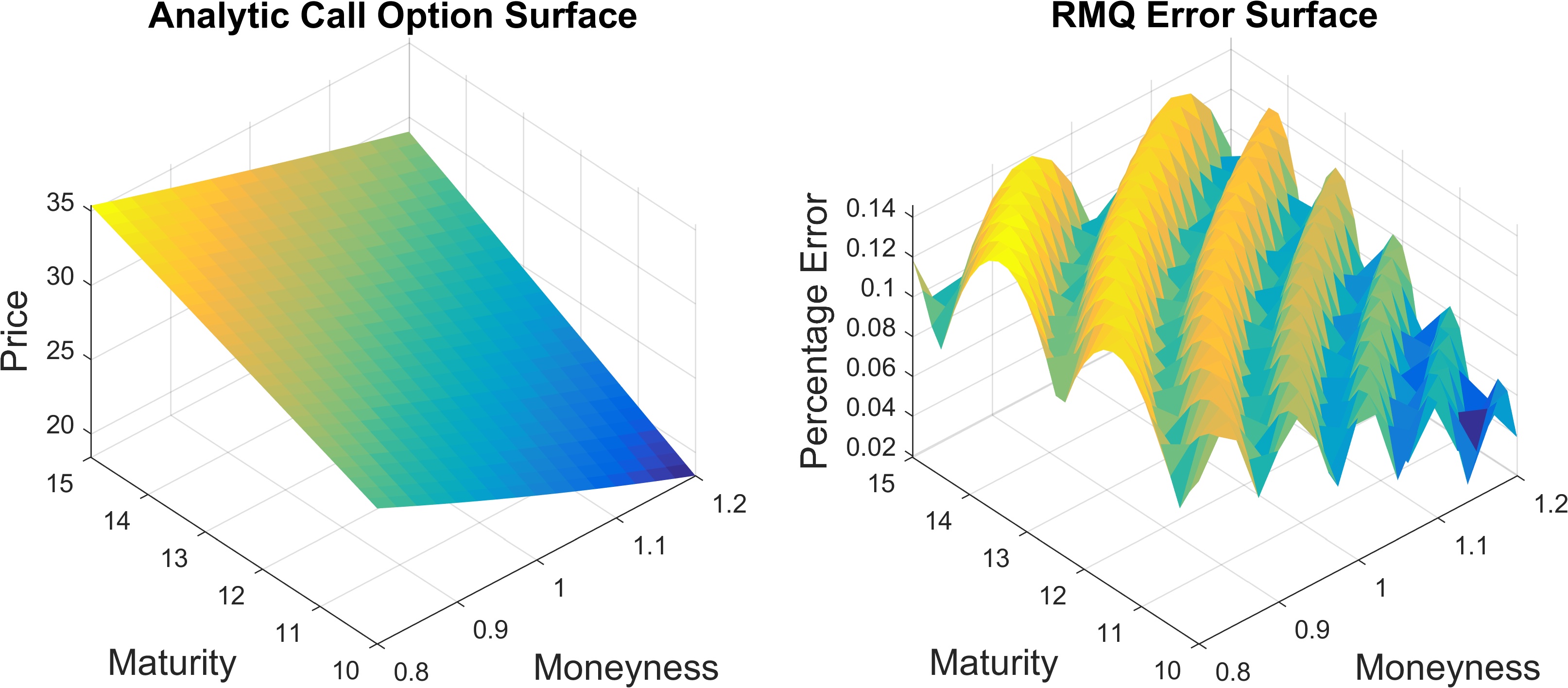}
    \end{center}
   \caption{European call option price surface with RMQ pricing error.}
    \label{Fig: European call options}
\end{figure}

For comparison, assume a hypothetical risk-neutral measure, $\P_\theta$, with the discounted GOP dynamics under this measure given by \eqref{Eq: RN dGOP}. Although this measure is not equivalent to $\P$, when the strike $K>0$, the risk-neutral call price, denoted $c^{\mathrm{RN}}_{T, K}$, corresponds to the real-world call price given above for $\delta_\theta = \tfrac{1 - 2a}{1 - a}$. This occurs because the measures differ only for values of $\GOP$ around $0$, which the integral in the call pricing problem avoids for positive strikes. However, the risk-neutral put option, denoted $p^{\mathrm{RN}}_{T, K}$, is significantly more expensive than the real-world put option for long time horizons.

To see this, consider the mathematical basis for put-call parity,
\begin{equation}
(K - \GOP_T)^+ = (\GOP_T - K)^+ - \GOP_T + K. \label{Eq: Put Call Basis}
\end{equation}
Taking the expectation under the real-world measure, with the GOP as the numeraire, provides the \emph{fair} put-call parity relationship,
\begin{equation*}
    p_{T,K}(t, \GOP_t) = c_{T, K}(t, \GOP_t) - \GOP_t + K P_T(t, \GOP_t),
\end{equation*}
where the fair or real-world zero coupon bond price is given by
\begin{equation*}
    P_T(t, \GOP_t) = \E*{\frac{\GOP_t}{\GOP_T}\Bigr\rvert \Ainfo_t} = \frac{\beta(t)}{\beta{(T)}}\E*{\frac{\dGOP_t}{\dGOP_T} \Bigr\rvert \Ainfo_t} = \frac{\beta(t)}{\beta{(T)}}\chi'^{\,2}\left(\frac{(\dGOP_t)^{2(1 - a)}}{\Delta \varphi}; \delta - 2, 0\right).
\end{equation*}
Of course, taking the discounted risk-neutral expectation of \eqref{Eq: Put Call Basis} provides the classical put-call parity relationship,
\begin{equation*}
    p^{\mathrm{RN}}_{T,K}(t, \GOP_t) = c^{\mathrm{RN}}_{T, K}(t, \GOP_t) - \GOP_t + K \frac{\beta(t)}{\beta{(T)}}.
\end{equation*}
Since the hypothetical risk-neutral call prices and real-world call prices coincide, and
\[ P_T(t, \GOP_t) \leq \frac{\beta(t)}{\beta{(T)}}, \]
the real-world put option prices must be less than or equal to the risk-neutral put option prices.

Figure \ref{Fig: European put options} illustrates the difference between long-dated at-the-money put options priced using the classical risk-neutral pricing theory and the prices provided by Proposition \ref{Prop: Put Option}, obtained using real-world pricing. The parameters used are taken from \cite{baldeaux2015hybrid}, where they were estimated from empirical data, with $\alpha_0 = 51.34,\ \eta=0.1239,\ c=0.1010$ and $ a = 0.2868$. The initial discounted GOP was set at $50$ and the constant short rate at $5\%$. Maturities are set at bi-monthly intervals from $5$ out to $15$ years.

The left panel of Figure \ref{Fig: European put options} shows how the prices correspond for short maturities, with the risk-neutral put becoming more and more expensive as the maturities lengthen. The right panel of Figure \ref{Fig: European put options} shows the difference between the risk-neutral put and the real-world put as a percentage of the classical risk-neutral option price. At a maturity of $15$ years, the real-world put option is $70\%$ less expensive to purchase. 

As a first example of RMQ, Figure \ref{Fig: European call options} shows an analytic European call option pricing surface along with the RMQ pricing error.  Moneyness is varied by changing the strike over the initial GOP value, and maturities are set at monthly intervals from $10$ to $15$ years.

The weak order 2.0 RMQ scheme was used \citep{mcwalter2018} with $12$ time-steps per year and $50$ codewords held constant across time. The final error is under $0.15\%$ irrespective of maturity and moneyness. The error oscillates across moneyness, as is expected for a tree-type method. Computing the RMQ grid out to $15$ years takes less than $1$ second.

\begin{figure}
    \begin{center}
        \includegraphics[width=\columnwidth]{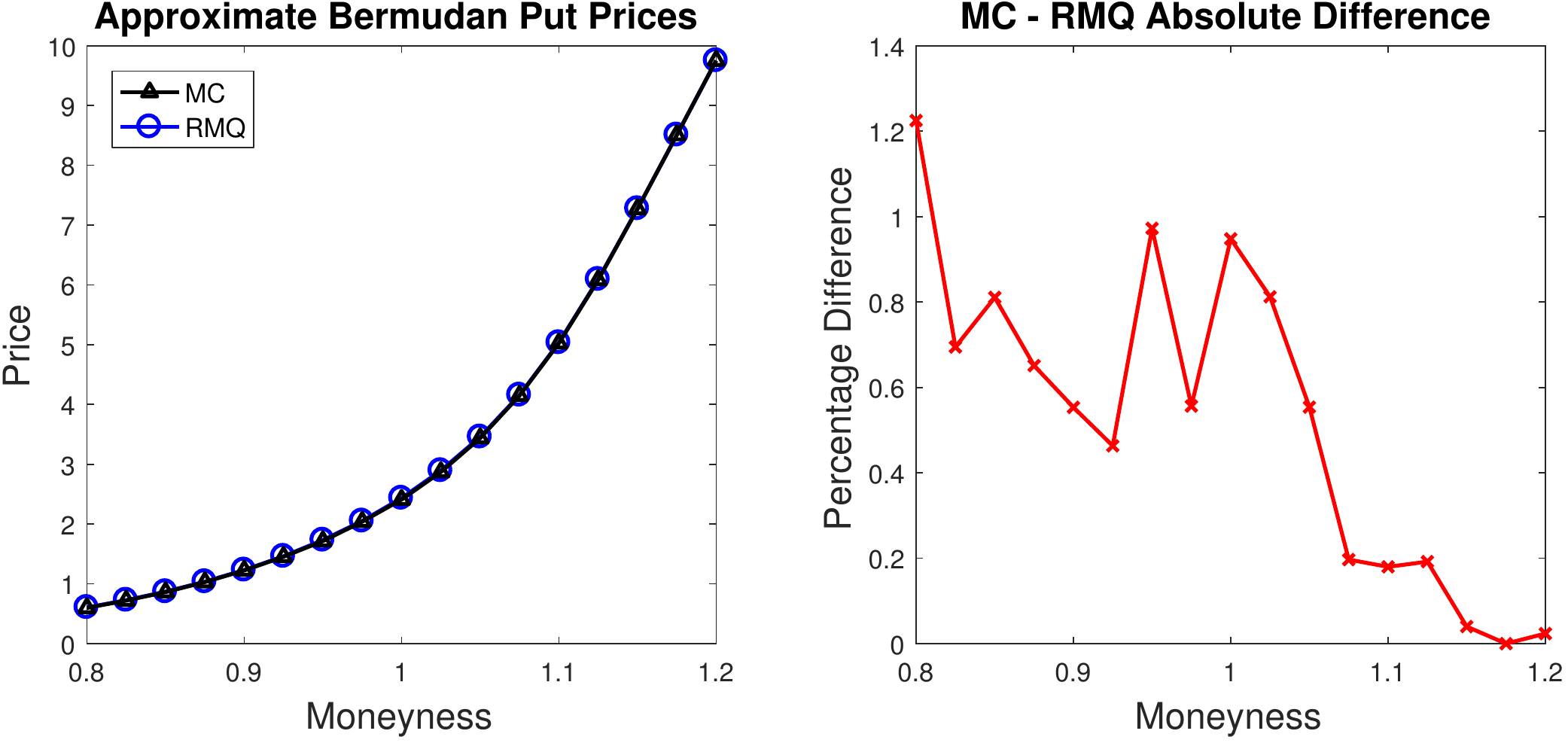}
    \end{center}
   \caption{Bermudan put options priced using least-squares Monte Carlo and RMQ.}
    \label{Fig: Bermudan put options}
\end{figure}

In Figure \ref{Fig: Bermudan put options}, Bermudan put options on the GOP with a maturity of $5$ years and monthly exercise opportunities are priced using both a least-squares Monte Carlo simulation and RMQ. The Monte Carlo simulation is a $500\,000$ path long-step simulation, using the exact transition density of the discounted GOP. The RMQ algorithm is again the weak order 2.0 scheme with $12$ time-steps and $100$ codewords. The Monte Carlo algorithm takes approximately $14.9$ seconds per strike, whereas the RMQ algorithm takes only $0.5$ seconds to price all strikes. The maximum difference between the two methods is $1.4\%$ of the price in the worst case. Thus, the methods agree very well across strikes with the RMQ algorithm being significantly faster to compute. Note that the RMQ results may well be more accurate than the least-squares Monte Carlo results for low moneyness, as the Monte Carlo simulation may be unreliable for deep out-the-money options.

\subsection{Stochastic Short Rates}
\begin{figure}
    \begin{center}
        \includegraphics[width=\columnwidth]{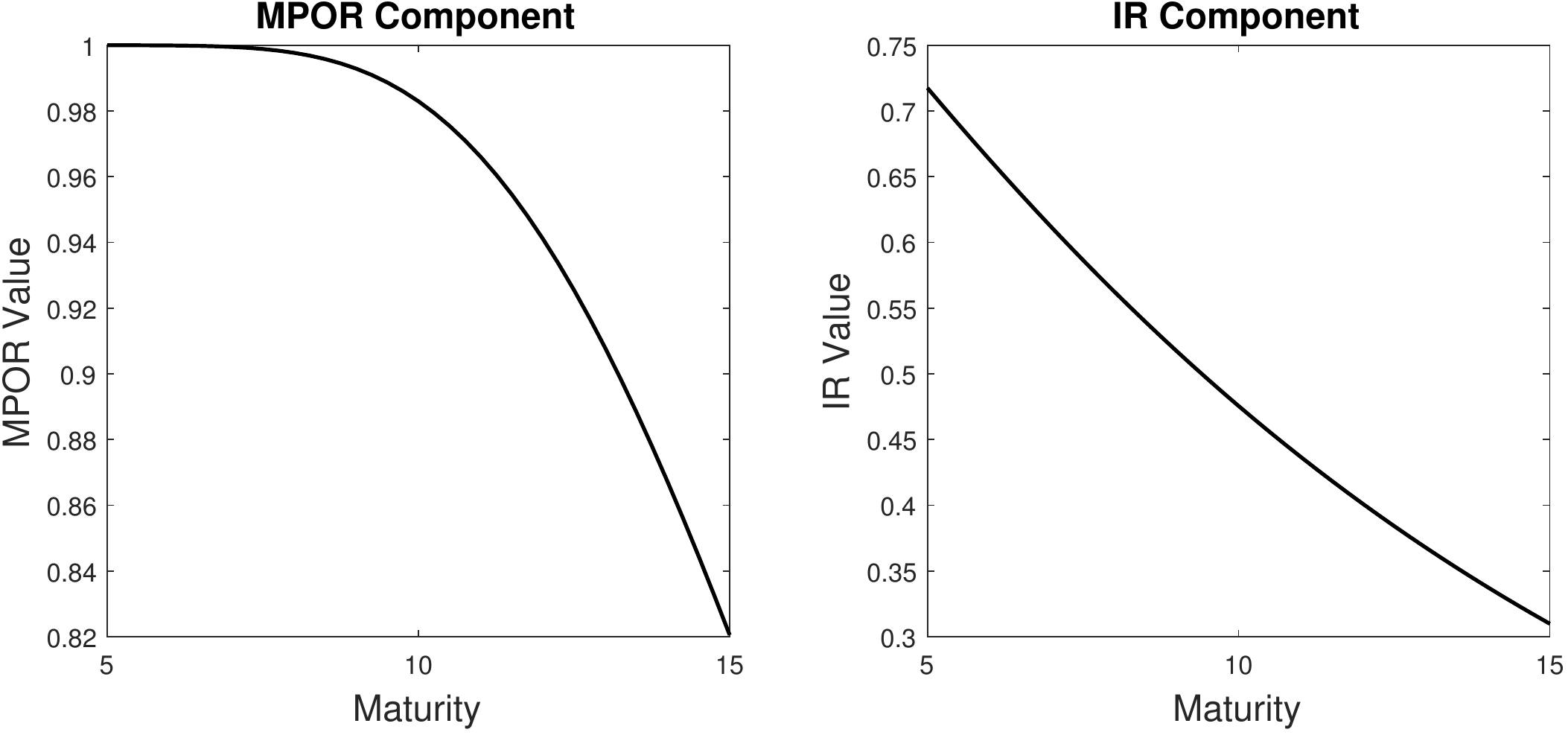}
    \end{center}
   \caption{The analytic MPOR and IR components of the fair zero-coupon bond price under the hybrid model.}
    \label{Fig: MPOR IR}
\end{figure}

In this subsection, the hybrid model, investigated by \cite{baldeaux2015hybrid}, is extended by relaxing the assumption of independence between the short rate and the discounted GOP. 
It is shown how the real-world zero-coupon bond prices become significantly less than risk-neutral prices as maturities increase. 
Fast and accurate numerical pricing for European put options on zero-coupon bonds is also provided.

\cite{baldeaux2015hybrid} undertake an empirical investigation to determine which short-rate model, combined with the TCEV model for the discounted GOP, performs best when pricing and hedging long-dated zero-coupon bonds. Their investigation concluded that the $3/2$ short-rate model of \cite{ahn1999parametric} outperforms the competing models in terms of capturing the dynamics of the real-world short-term interest rate, as well as delivering the smallest prices for zero-coupon bonds. 

Under the real-world measure, the $3/2$ short-rate model is described by
\begin{align}
dr_t &= \kappa(\theta r_t - r_t^2)\, dt + \sigma r_t^{3/2}\, dW^r_t, \label{Eq: 3/2 Short Rate Model} \end{align}
with $r_0 > 0$. Proposition \ref{Prop: ZCB} below is reproduced from \cite{baldeaux2015hybrid}.

\begin{prop}
\label{Prop: ZCB}
If the Brownian motion, $W^r$, driving the short rate is independent of the Brownian motion, $W$, driving the GOP then the time-$t$ price of a fair zero-coupon bond maturing at $T$ is given by
\begin{align}
P_T(t, r_t, \dGOP_t) &= \E*{\frac{\GOP_t}{\GOP_T} \Bigl|{\Ainfo_t}} = \E*{\frac{\dGOP_t}{\dGOP_T}\frac{S^0_t}{S^0_T} \Bigl|{\Ainfo_t}} = M(\dGOP_t, t, T)G(r_t, t, T),\label{Eq: ZCB Expression}\\
\intertext{with}
M(\dGOP_t, t, T) &= \E*{\frac{\dGOP_t}{\dGOP_T}\Bigl|{\Ainfo_t}} = \chi'^{\,2}\left(\frac{(\dGOP_t)^{2(1 - a)}}{\Delta \varphi}; \delta - 2, 0\right) \label{Eq: MPOR component} \\
\intertext{denoting the market price of risk component and}
G(r_t, t, T) &= \E*{\frac{S^0_t}{S^0_T} \Bigl|{\Ainfo_t}} = \E*{\exp\left({\int_t^Tr_s\, ds}\right) \Bigl|{\Ainfo_t}} \label{Eq: IR component}
\end{align}
denoting the interest-rate (IR) component.
\end{prop}
\begin{proof}
The independence of the GOP from the short rate, and thus the savings account, allows the expectation in \eqref{Eq: ZCB Expression} to be separated into the product of \eqref{Eq: MPOR component} and \eqref{Eq: IR component}. The right-hand side of \eqref{Eq: MPOR component} follows directly from the known transition density of the discounted GOP, provided by Proposition \ref{Prop: DGOP to BESQ}, and the right-hand side of \eqref{Eq: IR component} follows from the definition of the savings account, \eqref{Eq: Savings Account}.
\end{proof}

As a result of the above proposition, if the expectation in \eqref{Eq: IR component} was taken under the risk-neutral measure, the fair zero-coupon bond price could be interpreted as the product of the traditional risk-neutral bond price, $G(r_t, t, T)$, and the market price of risk component,  $M(\dGOP_t, t, T)$.
Note that the MPOR component is given explicitly as the probability that all paths of the inverse discounted GOP process have not attained zero by time $T$. This probability goes to $0$ as $T$ goes to infinity; eventually the growth-optimal portfolio dominates any other traded asset, ensuring that the expected value of the asset, expressed in terms of the GOP, goes to zero.

This behaviour can be seen in the left panel of Figure~\ref{Fig: MPOR IR}, where the MPOR component is plotted out to $15$ years using the model parameters in the previous section. Note that the MPOR component only becomes significantly less than one well after the $5$-year mark. This indicates that theoretical real-world bond prices and those obtained using classical risk-neutral pricing would coincide for maturities out to roughly $8$ years for these parameters. 

\begin{figure}
    \begin{center}
        \includegraphics[width=\columnwidth]{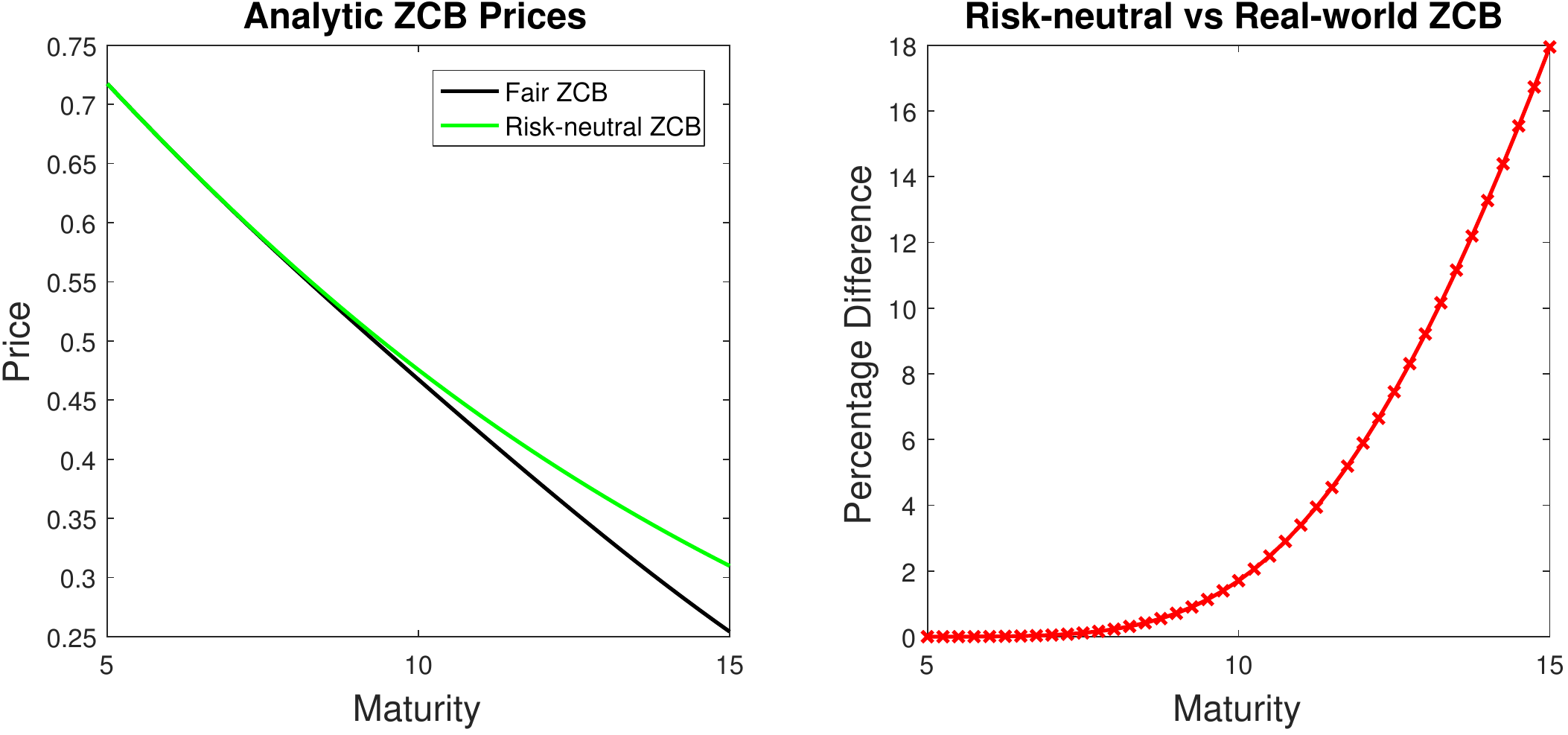}
    \end{center}
   \caption{The analytic fair zero-coupon bond price in the hybrid model compared to the hypothetical risk-neutral bond price for maturities ranging from $5$ to $15$ years.}
    \label{Fig: RN vs RW ZCBS}
\end{figure}

\begin{prop}
\label{Prop: IR Component}
Under the $3/2$ short-rate model,
\begin{equation*}
G(r_t, t, T) = \frac{\Gamma(\alpha - \gamma)x(r_t,t,T)^\gamma}{\Gamma(\alpha)}\prescript{}{1}F_1(\gamma, \alpha, -x(r_t, t, T))
\end{equation*}
with
\begin{align*}
\alpha&= \frac{2}{\sigma^2}\left(\kappa + (1 + \gamma)\sigma^2 \right),& \gamma&= \frac{1}{\sigma^2}\left(\sqrt{\phi^2 + 2\sigma^2} - \phi\right)\\
x(r, t, T) &= \frac{2\kappa\theta}{\sigma^2\left(e^{\kappa\theta(T - t)} - 1 \right)r},& \phi &= \kappa + \frac{1}{2}\sigma^2,  
\end{align*}
and where $\prescript{}{1}F_1$ is the confluent hypergeometric function of the first kind, or Kummer's function.
\end{prop}
\begin{proof}
See \citet[Sec.~3]{ahn1999parametric}.
\end{proof}

\begin{figure}
    \begin{center}
        \includegraphics[width=\columnwidth]{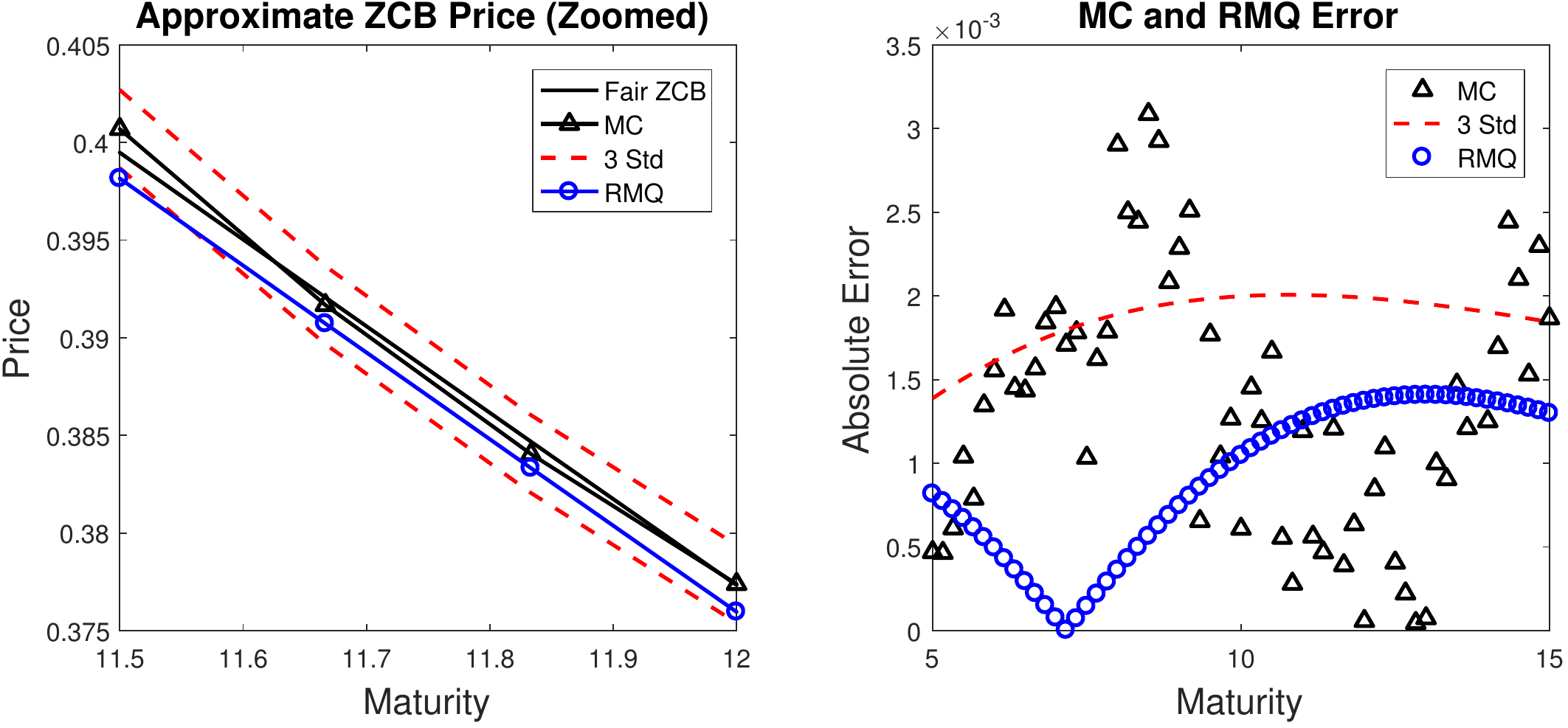}
    \end{center}
   \caption{Approximating the fair zero-coupon bond price using Monte Carlo simulation and RMQ.}
    \label{Fig: Hybrid ZCB}
\end{figure}

The IR component is illustrated in the right panel of Figure \ref{Fig: MPOR IR}, using the $3/2$ model parameters estimated from the market by \cite{baldeaux2015hybrid}, with $\kappa=3.5726,$ $\theta=0.096,$ $\sigma=0.7960$ and the initial short rate selected as $r_0=0.05$.

Finally, the analytic fair zero-coupon bond price is displayed in Figure \ref{Fig: RN vs RW ZCBS} for maturities out to $15$ years and contrasted with the risk-neutral bond price. 
The right panel of Figure \ref{Fig: RN vs RW ZCBS} shows the difference between the hypothetical risk-neutral bond and the real-world bond as a percentage of the classical risk-neutral bond price. At a maturity of $15$ years, the fair bond is $18\%$ less expensive to purchase. It is clear that the fair bond will continue to become less expensive, under the $3/2$ dynamics, as the maturity lengthens further. 
It is beyond the scope of this work to demonstrate the hedging of these contracts. \cite{hulley2012hedging} have, however, demonstrated that the theoretical real-world bond prices can be accurately hedged. 

In Figure \ref{Fig: Hybrid ZCB} the fair zero coupon bond price is approximated using Monte Carlo simulation and the RMQ algorithm.
The Monte Carlo simulation for the IR and MPOR components were each computed using $100\,000$ paths. 
The MPOR component could be long-stepped to each maturity considered, as the exact transition density is known, however the $3/2$ model was simulated using an Euler-Maruyama scheme with $6$ time steps per year. The RMQ algorithm also used $6$ time steps per year with $50$ codewords for the short-rate process and $150$ codewords for the discounted GOP process. 
The Monte Carlo simulation took $13.5$ seconds to compute, whereas the RMQ algorithm was more than twice as fast at $5.1$ seconds.

The absolute error for both methods is small; the left panel of Figure \ref{Fig: Hybrid ZCB} has been zoomed in to focus on a $6$-month period so that the difference between the approximations and the analytic value can be seen. The Monte Carlo method presents some bias over the full period, as more points lie outside the three standard deviation bounds than would be expected. The RMQ algorithm lies well within the error bounds for the full range of maturities.


\subsubsection{Affect of Short Rate and GOP Correlation}
Although the assumption of independence between the short rate and GOP seems restrictive, some empirical evidence is provided for it by \cite{baldeaux2015hybrid}. They use the daily $3$-month USD T-Bill rates as the proposed short rate and the EWI114 equi-weighted index to approximate the GOP. They find that the covariation remains close to zero and exhibits no clear trend. The theory of approximating the GOP using a well-diversified world-index is presented by \cite{platen2012approximating}. It is now demonstrated that the RMQ methodology can also efficiently handle the case of correlation between the short-rate and the GOP.

To account for correlation, the Joint RMQ algorithm, developed by \cite{rudd2017quantization}, has been used.
In the left panel of Figure \ref{Fig: Hybrid ZCB Correlation} the zero-coupon bond price under the hybrid model is displayed for a range of correlation values between $-0.9$ and $0.9$. The large range is chosen to exaggerate the effect. The impact of positive correlation is greater than that of negative correlation, but the overall impact of the correlation is small. Again, the figure has been zoomed in to focus on a $6$-month period. In the right panel of Figure \ref{Fig: Hybrid ZCB Correlation} the percentage difference between the zero-coupon bond price and the price with zero correlation is displayed for different correlations and maturities out to $15$ years. 

These numerical results substantiate the original finding: even for large values of the correlation the impact on the zero-coupon bond price is less than $0.8\%$. Thus, for bond-pricing applications, correlation between the GOP and the short rate may be safely neglected. Intuitively, the correlation can be viewed as affecting the path of the MPOR process while leaving the path of the IR process unchanged, for example, when running an Euler Monte Carlo simulation with the covariance matrix decomposed using the Cholesky decomposition. However, the path of the MPOR process only has a minimal effect on the zero-coupon bond price for shorter maturities, as seen in the left panel of Figure \ref{Fig: Hybrid ZCB}. Thus, changing the correlation only has a small effect on the final zero-coupon bond price.

\begin{figure}
    \begin{center}
        \includegraphics[width=\columnwidth]{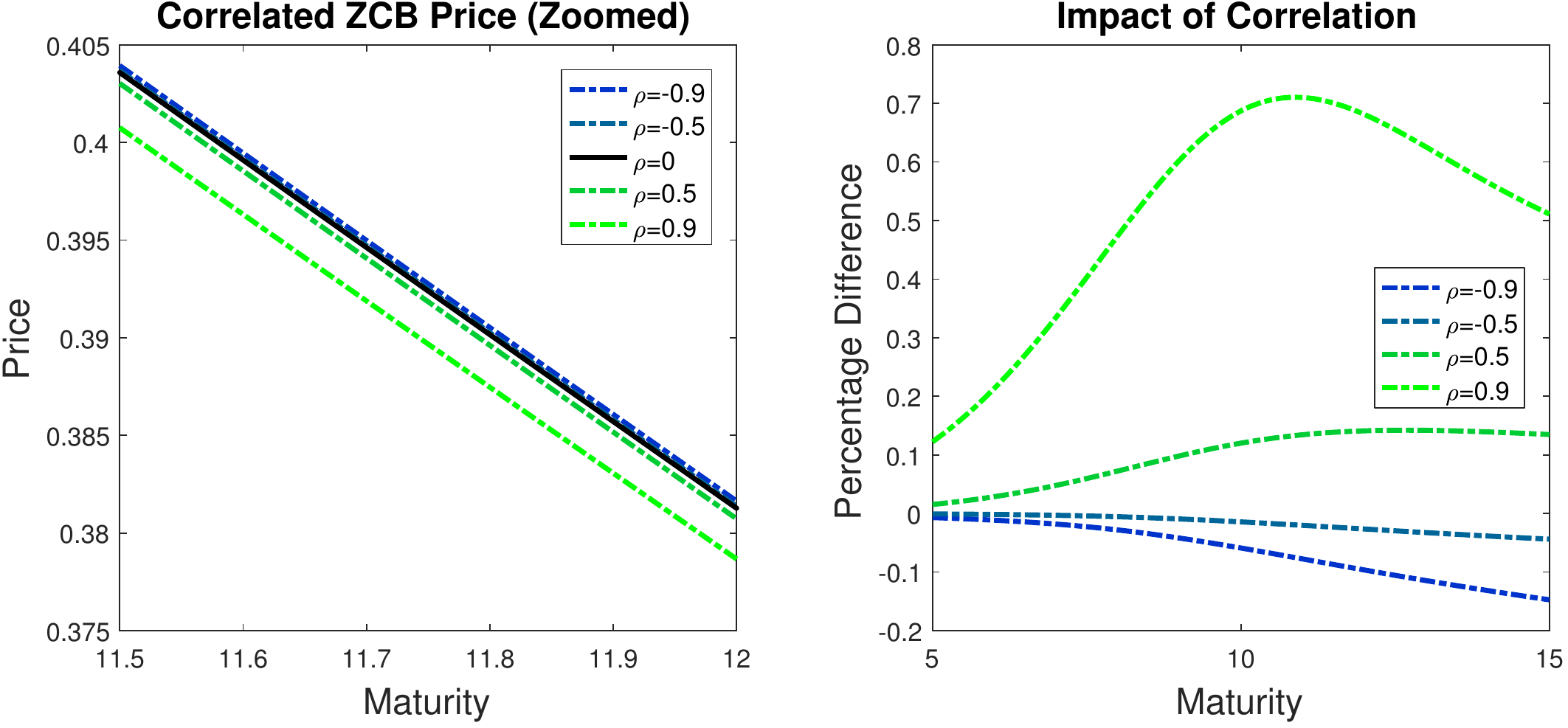}
    \end{center}
   \caption{The fair zero-coupon bond price in the hybrid model for a range of correlation values.}
    \label{Fig: Hybrid ZCB Correlation}
\end{figure}

\subsubsection{Zero-coupon Bond Options}
\begin{figure}
    \begin{center}
        \includegraphics[width=\columnwidth]{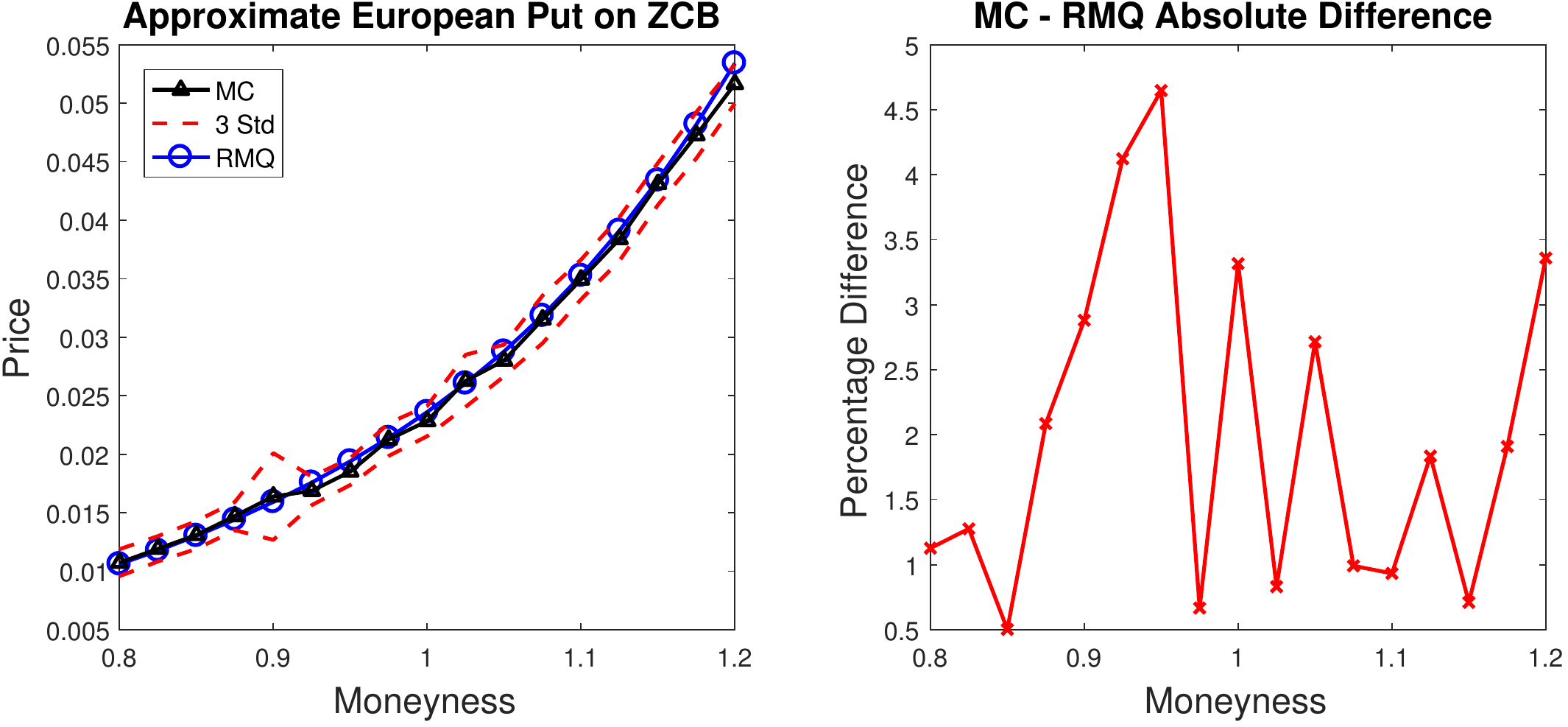}
    \end{center}
   \caption{Approximating the prices of a European put on a zero-coupon bond using Monte Carlo simulation and RMQ.}
    \label{Fig: Hybrid ZCO}
\end{figure}

To price a European put option on a zero-coupon bond, denoted ZCP, under the hybrid model with the option maturing at $T$ and the bond maturing at $S>T$, the following expectation must be computed
\[ \mathrm{ZCP}_{T,S,K}(t, r_t, \dGOP_t) \vcentcolon= \E*{\frac{\GOP_t}{\GOP_T}\left(K - P_S(T, r_T, \dGOP_T) \right)^+ \Bigl|{\Ainfo_t}}.\]
In Figure \ref{Fig: Hybrid ZCO}, the prices obtained using Monte Carlo simulation and the RMQ algorithm for the case $T=10$ years and $S=15$ years are displayed. The at-the-money strike is taken as the fair forward bond. 
As before, the Monte Carlo simulation for the IR and MPOR components used $100\,000$ paths each. The RMQ algorithm used $12$ time steps per year with $50$ codewords for the short-rate process and $150$ codewords for the discounted GOP process. 
The Monte Carlo simulation took $2.45$ seconds per strike, whereas the RMQ algorithm computed the prices for all strikes in $3.6$ seconds. 

There is no reference price available, but the prices obtained using the RMQ algorithm and Monte Carlo simulation lie sufficiently close together to indicate that RMQ is efficient and accurate. Barring a single deep in-the-money point, all the prices obtained using RMQ lie within the three standard deviation bounds of the Monte Carlo simulation. The average difference between the prices across all strikes is less than $2\%$.

It has already been established that, under the dynamics of the hybrid model, real-world zero-coupon bonds are less expensive than the bonds implied by traditional risk-neutral pricing. If a risk-neutral measure is assumed, this leads to an asymmetry in the prices of vanilla options on zero-coupon bonds. Real-world put options on fair bonds are more expensive than risk-neutral put options on risk-neutral bonds. The reverse is, of course, true for call options. This behaviour is illustrated in Figure \ref{Fig: RN vs RW ZCOs} for $T=5$ years and $S=10$ years. For the example depicted, the fair forward bond was computed using real-world pricing and the same strikes were used for both the real-world and risk-neutral options.

\begin{figure}
    \begin{center}
        \includegraphics[width=\columnwidth]{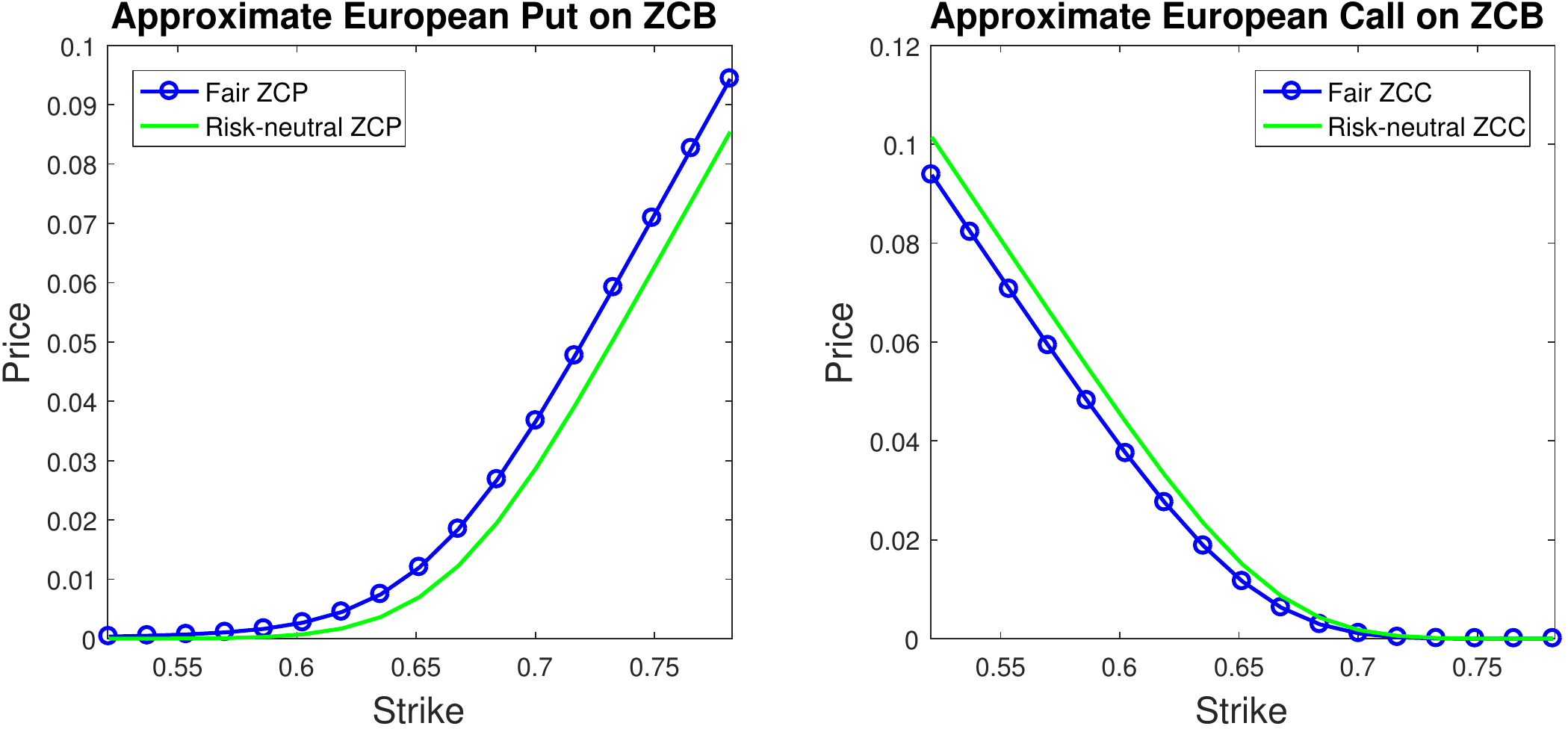}
    \end{center}
   \caption{Comparison of European put and call options written on zero-coupon bonds using real-world pricing and priced under a hypothetical risk-neutral measure.}
    \label{Fig: RN vs RW ZCOs}
\end{figure}

%% file: Sections/Conclusion.tex
\section{Conclusion}
\label{Sec: Conclusion}

In this paper, the benchmark approach was reviewed and a derivation of the real-world pricing theorem was presented for a two-asset continuous market.
It was demonstrated that real-world pricing may produce significantly lower prices for long-dated bonds and vanilla options than the classical risk-neutral pricing approach.
The time-dependent constant elasticity of variance model was used to model the growth-optimal portfolio. Under this model and the assumption of constant interest rates, analytic European option pricing formulae were derived in detail, extending the results of \cite{miller2008analytic,miller2010real}, for the modified constant elasticity of variance model and the stylized minimal market model. Recursive marginal quantization was used to efficiently and accurately produce long-dated European option pricing surfaces as well as price Bermudan options on the growth-optimal portfolio.

The hybrid model of \cite{baldeaux2015hybrid} is constructed by combining the TCEV model, for the growth-optimal portfolio, with the $3/2$ short-rate model of \cite{ahn1999parametric}. Under this combined model, RMQ was used to efficiently price long-dated zero-coupon bonds and options on zero-coupon bonds. The effect of introducing correlation between the growth-optimal portfolio and the stochastic short rate was investigated using the Joint RMQ algorithm, where it was shown that the correlation only has a minor impact.

This paper applied the RMQ algorithm outside the traditional risk-neutral framework by highlighting its effectiveness as a pricing mechanism for long-dated contracts under the benchmark approach. 

%% file: Sections/Appendix.tex

\titleformat{\section}{\normalfont\Large\bfseries}{Appendix~\thesection: }{0em}{}

\section{Squared Bessel Processes}
\label{App: Squared Bessel Processes}
This appendix summarizes properties of squared Bessel processes that are relevant to real-world pricing.

 Introduce $\vecW = \{\vecW_t = (w^1_t, w^2_t, \cdots, w^n_t)^\top, t\in[0,\infty)\}$, as an $n$-dimensional standard Brownian motion with $n\in\mathbb{N}$. Let the process $R = \{R_t = \norm{\vecW_t}, t\in[0,\infty)\},$ be the Euclidean norm of $\vecW$, such that $(R_t)^2 = \sum_{i=1}^{n}(w^i_t)^2$. It\^{o}'s formula provides
\begin{equation*}
    d(R_t)^2 = n\,dt + \sum_{i=1}^{n}2w^i_t\,dw^i_t.
\end{equation*}
For any $t>0$, $\P(R_t=0)=0$, such that
\begin{equation*}
    dW_t = \frac{1}{R_t}\sum_{i=1}^nw^i_t\,dw^i_t
\end{equation*}
is a real-valued Brownian motion via Lev\'y's characterization,\footnote{It is a continuous local martingale starting from zero and its quadratic variation is easily verified.} and $R$ satisfies
\begin{equation*}
    d(R_t)^2 = n\,dt + 2R_t\,dW_t.
\end{equation*}
Set $X_t = (R_t)^2$. Then, for every $\delta\in\mathbb{N}$ and $X_0 = x \geq 0$, the unique strong solution to the stochastic differential equation
\begin{equation*}
    dX_t = \delta\,dt + 2\sqrt{\abs{X_t}}\,dW_t,
\end{equation*}
is known as a \emph{squared Bessel process} of dimension $\delta$, denoted $\mathrm{BESQ}^\delta_t$. Although this intuitive derivation accounts only for squared Bessel processes of positive and integer dimension, this can be extended to $\delta\in\R$ \citep{revuz2013continuous}.


\begin{defn}[$\mathrm{BESQ}^\delta_t$]
    For every $\delta\in\R$ and $x\in\R$, the unique strong solution to
    \begin{equation}
        dX_t = \delta\, dt + 2\sqrt{\abs{X_t}}\,dW_t, \label{Eq: BESQ}
    \end{equation}
    with $X_0=x$, is called a squared Bessel process of dimension $\delta$, starting at $x$ and denoted by $\mathrm{BESQ}^\delta_t$.
\end{defn}

To relate the stochastic differential equations that arise when modelling the growth-optimal portfolio to squared Bessel processes, Proposition 6.3.1.1 of \citet{jeanblanc2009mathematical} is reproduced below. 

\begin{prop}
	\label{Prop: CIR to BESQ}
	Let $S = \{S_t,\ t\in[0,\infty)\}$ be a Cox-Ingersoll-Ross process satisfying
	\begin{equation*}
	    dS_t = \kappa(\theta-S_t)\,dt + \sigma\sqrt{S_t}\,dW_t,
	\end{equation*}
	with $S_0 = x \geq 0$ and $\kappa,\ \theta>0$, and define $\varphi(t) = \tfrac{\sigma^2}{4\kappa}(e^{\kappa t}-1).$
	Then
	\begin{equation*}
	    S_t = e^{-\kappa t}X_{\varphi(t)},
	\end{equation*}
	where $X_{\varphi(t)}$, $\varphi(t)\geq0$, is a $\textnormal{BESQ}^\delta_{\varphi(t)}$ process with dimension $\delta = \frac{4\kappa\theta}{\sigma^2}$.
\end{prop}

This allows the square-root process to be expressed as a time-transformed squared Bessel process, for which the transition density is well-understood. 



\subsection{The Transition Density of the Squared Bessel Process}
\label{Sec: BEQ Transition Density}
\cite{lindsay2012simulation} investigate the transition densities of squared Bessel processes under three different regimes, categorized by the dimension, $\delta$. They follow the classic analysis of \cite{feller1951two}, who proceeded by solving the Fokker-Planck equation associated with a more general version of \eqref{Eq: BESQ}. 


\subsubsection{\texorpdfstring{The case $\delta\leq 0$}{Case I}}
When $\delta\leq0$, the $X=0$ boundary is attainable and absorbing. The fundamental solution to the associated Fokker-Plank equation is the transition density
\begin{equation}
	p_{\delta\leq 0}(X_T, T; X_0) = \frac{1}{2T}\left(\frac{X_T}{X_0}\right)^{\frac{1}{2}\left(\frac{\delta}{2}-1\right)} \exp\left(- \frac{X_T+X_0}{2T}\right)I_{1 - \frac{\delta}{2}}\left(\frac{\sqrt{X_TX_0}}{T} \right), \label{Eq: Density dl2}
\end{equation}
where $I_\nu(x)$ is a modified Bessel function of the first kind with index $\nu$. By inspection, the above is related to the noncentral chi-squared density,
\begin{align}
	p_{\delta\leq 0}(X_T, T; X_0) &= \frac{1}{T} p_{\chi'^{\,2}}\left(\frac{X_0}{T}; 4 - \delta, \frac{X_T}{T} \right), \label{Eq: Chi2 Reversal} \\
	\intertext{expressed as a function of the noncentrality parameter, such that}
	\int_x^\infty p_{\delta\leq 0}(X, T; X_0)\, dX &= \int_x^\infty p_{\chi'^{\,2}}\left(\frac{X_0}{T}; 4 - \delta, \frac{X}{T} \right)\, dX \notag\\
	&= \chi'^{\,2}\left(\frac{X_0}{T}; 2-\delta, \frac{x}{T} \right). \label{Eq: Schroder}
\end{align}
The final step above was shown by \cite{schroder1989computing}.

This density is, however, \emph{norm-decreasing},
\begin{align}
\int_0^\infty p_{\delta\leq 0}(X, T; X_0)\, dX &= \chi'^{\,2}\left(\frac{X_0}{T}; 2-\delta, 0 \right)  \leq 1, \label{Eq: Staying Pos Prob}
\end{align}
as it does not include the probability of the process being absorbed at zero. 
\cite{lindsay2012simulation} propose constructing a full, norm-preserving density by adding a Dirac mass at the origin,
\begin{equation}
p^{\mathrm{full}}_{\delta\leq 0}(X_T, T; X_0) \vcentcolon= 2\left(1 - \chi'^{\,2}\left(\frac{X_0}{T}; 2-\delta, 0 \right)\right)\bar{\delta}(X_T) + p_{\delta\leq 0}(X_T, T; X_0),
\end{equation}
where $\bar{\delta}(x)$ is the Dirac delta function. Then the distribution of $X$ is given by
\begin{equation}
\P(X\leq X_T | X_0) = \int_0^{X_T} p^{\mathrm{full}}_{\delta\leq 0}(X, T; X_0)\, dX = 1 - \chi'^{\,2}\left(\frac{X_0}{T}; 2-\delta, \frac{X_T}{T} \right).
\end{equation}

\subsubsection{\texorpdfstring{The case $0<\delta<2$}{Case II}}
When $0<\delta<2$, the $X=0$ boundary is attainable and can be either absorbing or reflecting. If an absorbing boundary is selected, the analysis of the previous section holds,
\[ p^{\mathrm{A}}_{0<\delta<2}(X_T, T; X_0)\vcentcolon= p^{\mathrm{full}}_{\delta\leq 0}(X_T, T; X_0).\]
If a reflecting boundary is selected, the sign of the index of the Bessel function in the density changes, 
\begin{align}
	p^{\mathrm{R}}_{0<\delta<2}(X_T, T; X_0) &= \frac{1}{2T}\left(\frac{X_T}{X_0}\right)^{\frac{1}{2}\left(\frac{\delta}{2}-1\right)} \exp\left(- \frac{X_T+X_0}{2T}\right)I_{\frac{\delta}{2}-1}\left(\frac{\sqrt{X_TX_0}}{T} \right) \label{Eq: Density dg2}\\
	&= \frac{1}{T}p_{\chi'^{\,2}}\left(\frac{X_T}{T}; \delta, \frac{X_0}{T} \right), \notag
\end{align}
such that it is directly related to the non-central chi-squared density without the reversal of the roles of $X_T$ and $X_0$ seen in the $\delta\leq0$ case in \eqref{Eq: Chi2 Reversal}. This density is clearly norm-preserving.

\subsubsection{\texorpdfstring{The case $\delta>2$}{Case III}}
When $\delta>2$, the process can not attain zero and thus no boundary conditions can be specified. The transition density is of the same type as in the reflecting case,
\begin{equation}
p_{\delta> 2}(X_T, T; X_0) = \frac{1}{T}p_{\chi'^{\,2}}\left(\frac{X_T}{T}; \delta, \frac{X_0}{T} \right). \label{Eq: Reflecting Density}
\end{equation}

\subsubsection{Symmetry Relationships}
When $\delta>2$, the transition density is given by \eqref{Eq: Reflecting Density}. When considering only an absorbing boundary, for all $\delta<2$, the norm-decreasing part of the transition density is given by \eqref{Eq: Density dl2}. A simple symmetry relationship exists between these two expressions, 
\begin{equation}
	p_{\delta}(X_T, T; X_0) = p_{4-\delta}(X_0, T; X_T), \label{Eq: Density transformation}
\end{equation}
that aids in the option pricing problems considered in the paper.

Furthermore, 
\begin{align}
	X_T^{\left(1 - \frac{\delta}{2}\right)} p_{\delta}(X_T, T; X_0) = X_0^{\left(1 - \frac{\delta}{2}\right)} p_{\delta}(X_0, T; X_T) = X_0^{\left(1 - \frac{\delta}{2}\right)}p_{4-\delta}(X_T, T; X_0). \label{Eq: Symmetry}
\end{align}
The first equation follows from simple arithmetic using either \eqref{Eq: Density dg2} or \eqref{Eq: Density dl2} and the second equation follows from the first symmetry relationship above, \eqref{Eq: Density transformation}.

An important implication of the above result is that if $X$ is a $\mathrm{BESQ}^\delta$ process with $\delta>2$ and $X_0>0$, the process $Z_t = X_t^{\left(1 - \frac{\delta}{2}\right)}$ is a \emph{strict} local martingale\footnote{In fact it is a supermartingale, by Fatou's lemma, as it is bounded below}. This follows because the integral of the last term above is strictly less than one:
\begin{align*}
\E*{Z_T} &= \int_0^\infty X^{\left(1 - \frac{\delta}{2}\right)} p_{\delta>2}(X, T; X_0)\, dX \\
 &= \int_0^\infty X_0^{\left(1 - \frac{\delta}{2}\right)}p_{(4 - \delta)<2}(X_0, T; X)\, dX\\
&= X_0^{\left(1 - \frac{\delta}{2}\right)} \chi'^{\,2}\left(\frac{X_0}{T}; \delta-2, 0 \right) \\
&< Z_0.
\end{align*}

%% file: Real_World_Option_Pricing.bbl
\begin{thebibliography}{32}
\providecommand{\natexlab}[1]{#1}
\providecommand{\url}[1]{\texttt{#1}}
\expandafter\ifx\csname urlstyle\endcsname\relax
  \providecommand{\doi}[1]{doi: #1}\else
  \providecommand{\doi}{doi: \begingroup \urlstyle{rm}\Url}\fi

\bibitem[Ahn and Gao(1999)]{ahn1999parametric}
D.~H. Ahn and B.~Gao.
\newblock A parametric nonlinear model of term structure dynamics.
\newblock \emph{The Review of Financial Studies}, 12\penalty0 (4):\penalty0
  721--762, 1999.

\bibitem[Baldeaux et~al.(2014)Baldeaux, Ignatieva, and
  Platen]{baldeaux2014tractable}
J.~Baldeaux, K.~Ignatieva, and E.~Platen.
\newblock A tractable model for indices approximating the growth optimal
  portfolio.
\newblock \emph{Studies in Nonlinear Dynamics and Econometrics}, 18\penalty0
  (1):\penalty0 1--21, 2014.

\bibitem[Baldeaux et~al.(2015)Baldeaux, Fung, Ignatieva, and
  Platen]{baldeaux2015hybrid}
J.~Baldeaux, M.~C. Fung, K.~Ignatieva, and E.~Platen.
\newblock A hybrid model for pricing and hedging of long-dated bonds.
\newblock \emph{Applied Mathematical Finance}, 22\penalty0 (4):\penalty0
  366--398, 2015.

\bibitem[Delbaen and Schachermayer(1994)]{delbaen1994general}
F.~Delbaen and W.~Schachermayer.
\newblock A general version of the fundamental theorem of asset pricing.
\newblock \emph{Mathematische {A}nnalen}, 300\penalty0 (1):\penalty0 463--520,
  1994.

\bibitem[Delbaen and Schachermayer(1998)]{delbaen1998fundamental}
F.~Delbaen and W.~Schachermayer.
\newblock The fundamental theorem of asset pricing for unbounded stochastic
  processes.
\newblock \emph{Mathematische {A}nnalen}, 312\penalty0 (2):\penalty0 215--250,
  1998.

\bibitem[Feller(1951)]{feller1951two}
W.~Feller.
\newblock Two singular diffusion problems.
\newblock \emph{Annals of {M}athematics}, 54\penalty0 (1):\penalty0 173--182,
  1951.

\bibitem[Fontana(2015)]{fontana2015weak}
C.~Fontana.
\newblock Weak and strong no-arbitrage conditions for continuous financial
  markets.
\newblock \emph{International Journal of Theoretical and Applied Finance},
  18\penalty0 (01):\penalty0 1550005, 2015.

\bibitem[Geman et~al.(1995)Geman, El~Karoui, and Rochet]{geman1995changes}
H.~Geman, N.~El~Karoui, and J.-C. Rochet.
\newblock Changes of numeraire, changes of probability measure and option
  pricing.
\newblock \emph{Journal of Applied probability}, 32\penalty0 (2):\penalty0
  443--458, 1995.

\bibitem[Hakansson and Ziemba(1995)]{hakansson1995capital}
N.~H. Hakansson and W.~T. Ziemba.
\newblock Capital growth theory.
\newblock \emph{Handbooks in Operations Research and Management Science},
  9:\penalty0 65--86, 1995.

\bibitem[Heath and Platen(2002)]{heath2002consistent}
D.~Heath and E.~Platen.
\newblock Consistent pricing and hedging for a modified constant elasticity of
  variance model.
\newblock \emph{Quantitative Finance}, 2\penalty0 (6):\penalty0 459--467, 2002.

\bibitem[Hulley and Platen(2012)]{hulley2012hedging}
H.~Hulley and E.~Platen.
\newblock Hedging for the long run.
\newblock \emph{Mathematics and Financial Economics}, 6\penalty0 (2):\penalty0
  105--124, 2012.

\bibitem[Hunt and Kennedy(2004)]{hunt2004financial}
P.~Hunt and J.~Kennedy.
\newblock \emph{Financial {D}erivatives in {T}heory and {P}ractice}.
\newblock John Wiley \& Sons, 2004.

\bibitem[Jeanblanc et~al.(2009)Jeanblanc, Yor, and
  Chesney]{jeanblanc2009mathematical}
M.~Jeanblanc, M.~Yor, and M.~Chesney.
\newblock \emph{Mathematical {M}ethods for {F}inancial {M}arkets}.
\newblock Springer, 2009.

\bibitem[Karatzas and Kardaras(2007)]{karatzas2007numeraire}
I.~Karatzas and C.~Kardaras.
\newblock The numeraire portfolio in semimartingale financial models.
\newblock \emph{Finance and Stochastics}, 11\penalty0 (4):\penalty0 447--493,
  2007.

\bibitem[Kelly(1956)]{kelly1956new}
J.~Kelly.
\newblock A new interpretation of information rate.
\newblock \emph{IRE Transactions on Information Theory}, 3\penalty0
  (2):\penalty0 185--189, 1956.

\bibitem[Lindsay and Brecher(2012)]{lindsay2012simulation}
A.~Lindsay and D.~Brecher.
\newblock Simulation of the {CEV} process and the local martingale property.
\newblock \emph{Mathematics and Computers in Simulation}, 82\penalty0
  (5):\penalty0 868--878, 2012.

\bibitem[Long(1990)]{long1990numeraire}
J.~B. Long.
\newblock The numeraire portfolio.
\newblock \emph{Journal of {F}inancial {E}conomics}, 26\penalty0 (1):\penalty0
  29--69, 1990.

\bibitem[McWalter et~al.(2018)McWalter, Rudd, Kienitz, and
  Platen]{mcwalter2018}
T.~A. McWalter, R.~Rudd, J.~Kienitz, and E.~Platen.
\newblock Recursive marginal quantization of higher-order schemes.
\newblock \emph{Quantitative Finance}, 2018.
\newblock \doi{10.1080/14697688.2017.1402125}.
\newblock URL \url{https://doi.org/10.1080/14697688.2017.1402125}.

\bibitem[Miller and Platen(2008)]{miller2008analytic}
S.~M. Miller and E.~Platen.
\newblock Analytic pricing of contingent claims under the real-world measure.
\newblock \emph{International Journal of Theoretical and Applied Finance},
  11\penalty0 (08):\penalty0 841--867, 2008.

\bibitem[Miller and Platen(2010)]{miller2010real}
S.~M. Miller and E.~Platen.
\newblock Real-world pricing for a modified constant elasticity of variance
  model.
\newblock \emph{Applied Mathematical Finance}, 17\penalty0 (2):\penalty0
  147--175, 2010.

\bibitem[Pag{\`e}s and Sagna(2015)]{pages2015recursive}
G.~Pag{\`e}s and A.~Sagna.
\newblock Recursive marginal quantization of the {E}uler scheme of a diffusion
  process.
\newblock \emph{Applied Mathematical Finance}, 22\penalty0 (5):\penalty0
  463--498, 2015.

\bibitem[Platen(2001)]{Platen2001minimal}
E.~Platen.
\newblock A minimal financial market model.
\newblock Discussion Papers, Interdisciplinary Research Project 373:
  Quantification and Simulation of Economic Processes 2000,91, Berlin, 2001.
\newblock URL \url{http://hdl.handle.net/10419/62176}.
\newblock urn:nbn:de:kobv:11-10048178.

\bibitem[Platen(2002)]{platen2002arbitrage}
E.~Platen.
\newblock Arbitrage in continuous complete markets.
\newblock \emph{Advances in Applied Probability}, 34\penalty0 (3):\penalty0
  540--558, 2002.

\bibitem[Platen(2004)]{platen2004pricing}
E.~Platen.
\newblock Pricing and hedging for incomplete jump diffusion benchmark models.
\newblock In \emph{AMS-IMS-SIAM Joint Summer Research Conference on Mathematics
  of Finance}. American Mathematical Society, 2004.

\bibitem[Platen(2006)]{platen2004benchmark}
E.~Platen.
\newblock A benchmark approach to finance.
\newblock \emph{Mathematical Finance}, 16\penalty0 (1):\penalty0 131--151,
  2006.

\bibitem[Platen(2008)]{platen2008law}
E.~Platen.
\newblock Law of the minimal price.
\newblock Technical report, University of Technology, Sydney. QFRC Research
  Paper 215, 2008.

\bibitem[Platen and Heath(2006)]{platen2006benchmark}
E.~Platen and D.~Heath.
\newblock \emph{A {B}enchmark {A}pproach to {Q}uantitative {F}inance}.
\newblock Springer, 2006.

\bibitem[Platen and Rendek(2012)]{platen2012approximating}
E.~Platen and R.~Rendek.
\newblock Approximating the numeraire portfolio by naive diversification.
\newblock \emph{Journal of Asset Management}, 13\penalty0 (1):\penalty0 34--50,
  2012.

\bibitem[Revuz and Yor(1999)]{revuz2013continuous}
D.~Revuz and M.~Yor.
\newblock \emph{Continuous {M}artingales and {B}rownian {M}otion}.
\newblock Springer, 1999.

\bibitem[Rudd et~al.(2017)Rudd, McWalter, Kienitz, and
  Platen]{rudd2017quantization}
R.~Rudd, T.~A. McWalter, J.~Kienitz, and E.~Platen.
\newblock Fast quantization of stochastic volatility models.
\newblock \emph{Available at SSRN 2956168}, 2017.

\bibitem[Samuelson(1977)]{samuelson1977st}
P.~A. Samuelson.
\newblock St. petersburg paradoxes: Defanged, dissected, and historically
  described.
\newblock \emph{Journal of Economic Literature}, 15\penalty0 (1):\penalty0
  24--55, 1977.

\bibitem[Schroder(1989)]{schroder1989computing}
M.~Schroder.
\newblock Computing the constant elasticity of variance option pricing formula.
\newblock \emph{Journal of Finance}, 44\penalty0 (1):\penalty0 211--219, 1989.

\end{thebibliography}
